%% file: main_arxiv.tex
\documentclass[journal,onecolumn,12pt]{IEEEtran}
\usepackage{xcolor}

\usepackage{bbm}
\let\emph\textit
\usepackage[hidelinks]{hyperref}
\usepackage[ruled]{algorithm2e}
\usepackage{amsmath,amssymb,amsfonts,amsbsy} 
\usepackage{amsthm}
\usepackage{graphicx}
\graphicspath{{Figs/} }

\usepackage{cite}
\usepackage{textcomp}
\DeclareMathAlphabet\mathbfcal{OMS}{cmsy}{b}{n} 
\newtheorem{assumption}{Assumption}
\newtheorem{theorem}{Theorem}

\newtheorem{lemma}{Lemma}

\def\BibTeX{{\rm B\kern-.05em{\sc i\kern-.025em b}\kern-.08em
    T\kern-.1667em\lower.7ex\hbox{E}\kern-.125emX}}
\SetKwIF{If}{ElseIf}{Else}{if}{}{else if}{else}{end if}%

\begin{document}
\title{Quantum Conformal Prediction for\\ Reliable Uncertainty Quantification in\\ Quantum Machine Learning}
\author{Sangwoo Park,~\IEEEmembership{Member,~IEEE,} and 
        Osvaldo Simeone,~\IEEEmembership{Fellow,~IEEE}%
\thanks{Code can be found at \url{https://github.com/kclip/quantum-CP}.}
\thanks{{\color{black}The authors are with the King’s Communications, Learning \& Information Processing (KCLIP) lab within the Centre for Intelligent Information Processing Systems (CIIPS), Department of Engineering, King’s College London, London WC2R 2LS, U.K.  (e-mail: sangwoo.park@kcl.ac.uk; osvaldo.simeone@kcl.ac.uk).}}
\thanks{{\color{black}This work was supported by the European Research Council (ERC) under the
European Union’s Horizon 2020 Research and Innovation Programme (grant agreement No. 725732), by the European Union’s Horizon Europe project CENTRIC (101096379),  by an Open Fellowship of the EPSRC (EP/W024101/1),  by the EPSRC project (EP/X011852/1), and by  Project REASON, a UK Government funded project under the Future Open Networks Research Challenge (FONRC) sponsored by the Department of Science Innovation and Technology (DSIT).}}}

\pagenumbering{arabic}

\maketitle
\thispagestyle{plain}
\pagestyle{plain}

\begin{abstract}
In this work, we aim at augmenting the decisions output by quantum models with  ``error bars'' that provide finite-sample coverage guarantees. Quantum models implement  implicit probabilistic predictors that produce multiple random decisions for each input through measurement shots. Randomness arises not only from the inherent stochasticity of quantum measurements, but also from quantum gate noise and quantum measurement noise caused by noisy hardware. Furthermore, quantum noise may be correlated across shots and it may present drifts in time. This paper proposes to leverage such randomness to define prediction sets for both classification and regression that provably capture the uncertainty of the model. The approach builds on probabilistic conformal prediction (PCP), while accounting for the unique features of quantum models. Among the key technical innovations, we introduce a new general class of non-conformity scores that address the presence of quantum noise, including possible drifts. Experimental results, using both simulators and current quantum computers, confirm the theoretical calibration guarantees of the proposed framework.
\end{abstract}


\begin{IEEEkeywords}
Quantum machine learning, conformal prediction, generalization analysis, uncertainty quantification. 
\end{IEEEkeywords}

\IEEEpeerreviewmaketitle

\input{Sections/1_Intro_AISTATS}

\input{Sections/2_CP_AISTATS}

\input{Sections/4_QCP_AISTATS.tex}

\input{Sections/6_Experimental_Settings}

\input{Sections/7_Experimental_Results}
\input{Sections/8_Conclusion}
\input{Sections/9_appendix}

\bibliographystyle{IEEEtran}
\bibliography{mybib}
\end{document}

%% file: Sections/1_Intro_AISTATS.tex
\section{Introduction}
\label{sec:intro}
Quantum machine learning (QML) is currently viewed as a promising paradigm for the optimization of algorithms that can leverage existing noisy intermediate scale quantum (NISQ) computers \cite{biamonte2017quantum,schuld2021machine, simeone2022introduction}. As for classical machine learning, in order for QML to be useful for sensitive decision-making application, it is essential to endow it with the capacity to reliably quantify uncertainty \cite{tran2022plex}.

\begin{figure}
\center
\includegraphics[width=0.75\columnwidth]{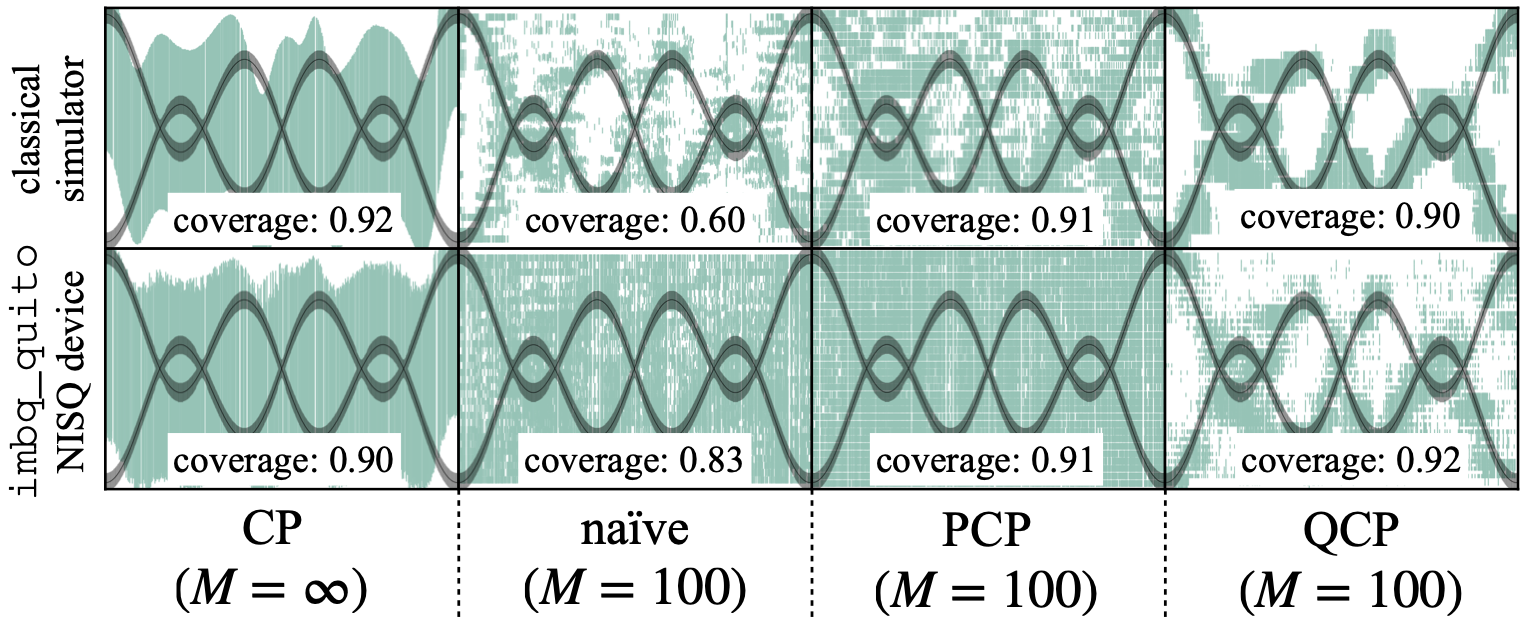}
  \caption{This paper addresses the problem of using $M$ samples drawn from quantum model to produce error bars  at coverage probability $1-\alpha$ for the target variable given a test input. The ground-truth (unknown) minimal $1-\alpha$ conditional support for a regression problem  is shown as the gray area. Predicted intervals are shown in green that are  produced by CP \cite{vovk2022algorithmic} applied to a conventional predictor based on the expected value over a large number $(M=20,000)$ of measurements shots (Sec.~\ref{subsec:CP_for_infinite_shot}); by a na\"ive predictor applied to $M=100$ shots (Sec.~\ref{sec:experimental_results});  by probabilistic CP \cite{wang2022probabilistic} applied to $M=100$ shots; and by QCP (this work), also applied to $M=100$ shots.} \label{fig:regression_learning_intro}
\end{figure}

This paper introduces a novel framework, referred to as \emph{quantum conformal prediction} (QCP), that augments QML models with provably reliable ``error bars''. QCP applies a post-hoc calibration step based on held-out calibration, or validation, data to a pre-trained quantum model, and it builds on \emph{conformal prediction} (CP), a general calibration methodology that is currently experiencing renewed attention in the area of classical machine learning \cite{vovk2022algorithmic,tibshirani2019conformal, barber2021predictive, angelopoulos2021gentle}. Unlike classical CP, QCP takes into account the unique nature of quantum models as \emph{probabilistic} machines, due to the inherent randomness of quantum measurements and to the presence of quantum hardware noise \cite{simeone2022introduction, cai2022quantum}.

\subsection{Quantum Circuits as Implicit Probabilistic Models} Quantum models implement  \emph{implicit} probabilistic predictors that produce multiple random decisions for each input through measurement shots (see, e.g., \cite[Sec.~3]{simeone2022introduction}). Randomness arises not only from the inherent stochasticity of quantum measurements, but also from quantum gate noise and quantum measurement noise caused by noisy hardware \cite{georgopoulos2021modeling, bravyi2021mitigating, smith2021qubit}. Furthermore, quantum noise may be \emph{correlated} across shots and it may present \emph{drifts} \cite{schwarz2011detecting, van2013quantum}. In this regard, while quantum error correction codes \cite{calderbank1996good} are too complex to run on NISQ hardware, \emph{quantum error mitigation} (QEM) is often implemented as a classical post-processing tool to mitigate the bias caused by quantum noise, while increasing the variance of the measurement outputs \cite{cai2022quantum, jose2022error}.

The goal of QCP is to assign \emph{well-calibrated} ``error bars'' to the decisions made by a pre-trained QML probabilistic model. The error bars are formally subsets of the output space, and \emph{calibration} refers to the property of containing the true target with probability no smaller than a predetermined \emph{coverage} level. We wish to ensure calibration guarantees that hold irrespective of the size of the training data set, of the ansatz of the QML model, of the training algorithm, of the number of shots, and of the type, correlation, and drift of  quantum hardware noise.

\subsection{Related Work} \label{subsec:related_work}
CP is a \emph{post-hoc} calibration methodology that yields a set  prediction, rather than a single point prediction, that contains the true target with predetermined coverage level \cite{vovk2022algorithmic}. Given a new test input, CP constructs a predicted set by collecting all possible outputs that \emph{conform} well with held-out  calibration data. Conformity is measured by a \emph{scoring function} that is typically defined based on the loss accrued by the pre-trained model on the given data point.

Recent work \cite{wang2022probabilistic} introduced a novel CP scheme, referred to as \emph{probabilistic CP} (PCP), that applies to \emph{classical} sampling-based  predictors. The main motivation of reference \cite{wang2022probabilistic} in leveraging probabilistic outputs is to address the mismatch between the assumed probabilistic model and the ground-truth distribution, which may require disconnected predictive sets.

QCP is inspired by reference \cite{wang2022probabilistic}, although our motivation stems from the fact that quantum models are \emph{inherently} implicit probabilistic models. As we explain next, to account for the unique features of QML probabilistic models, we introduce a new class of scoring functions that address the presence of quantum noise, including possible correlation and drifts \cite{schwarz2011detecting}.

Generalization analysis for QML also aims at quantifying test performance, although the focus is on determining scaling laws as a function of the size of the training set (see, e.g., \cite{banchi2023statistical} for an overview). Accordingly, as we briefly discuss in Appendix~\ref{supp:sec:gen_vs_cp}, quantum generalization analysis fails to provide operationally meaningful error bars, unlike QCP.

\subsection{Main Contributions}
Our main contributions are as follows.\\
    \noindent $\bullet$ We introduce QCP, a post-hoc calibration methodology for QML models that produces predictive sets with coverage guarantees that hold irrespective of the size of the training data set, of the ansatz of the QML model, of the training algorithm, of the number of shots, and of the type, correlation, and drift of  quantum hardware noise. \\
\noindent $\bullet$ We detail experimental results based on both simulation and quantum hardware implementations. The experiments confirm the  theoretical calibration guarantees of QCP, while also illustrating its merits in terms of informativeness of the predicted set (see Fig.~\ref{fig:regression_learning_intro} for a preview). 


%% file: Sections/2_CP_AISTATS.tex

\section{Conventional Conformal Prediction for Deterministic Quantum Models}
\label{subsec:quantum_models}
\label{sec:CCP}
In this section, we introduce, for reference, a direct application of conventional CP \cite{vovk2022algorithmic, angelopoulos2021gentle} to predictions obtained via QML models. To this end, we make the conventional assumption that the outputs of the quantum model are obtained by averaging over a large number of shots, as in most of the literature on the subject (see, e.g., \cite{schuld2021machine}). Probabilistic models are studied in the next section. Throughout,   bold fonts are used to denote random variables.

\subsection{Quantum Circuits as Deterministic Models} \label{sec:pqc_deterministic}
A   \emph{parameterized quantum circuit} (PQC) encodes the classical input $x \in \mathcal{X}$ into a \emph{parameterized quantum embedding} defined by the state of an $n$-qubits state (see, e.g., \cite{schuld2021machine, simeone2022introduction}).  The state is described by a $2^n\times 2^n$ density matrix $\rho(x|\theta)$, which is a positive semi-definite matrix with unitary trace, i.e., $\text{Tr}(\rho(x|\theta))=1$, where $\text{Tr}(\cdot)$ represents the trace operator. We write $N=2^n$ for the dimension of the Hilbert space on which density matrix $\rho(x|\theta)$ operates. The state $\rho(x|\theta)$ is obtained via the application of a parameterized unitary matrix $U(x|\theta)$ to a fiducial state $|0\rangle$ for the register of $n$ qubits, yielding \begin{align}\rho(x|\theta)=U(x|\theta)|0\rangle \langle 0| U(x|\theta)^\dagger,\label{eq:quantum_embedding}\end{align} where $\dagger$ represents the complex conjugate transpose operation.

Deterministic predictors can be in principle implemented via a PQC by considering the expected value of some \emph{observable} $O$ as the output of the model. An observable is defined by an $N\times N$ Hermitian matrix  $\label{eq:observable} O= \sum_{j=1}^{N'} o_j \Pi_j$ with real eigenvalues $\{o_j\}_{j=1}^{N'}$ and projection matrices $\{ \Pi_j \}_{j=1}^{N'}$, where $N' \leq N$ is the number of distinct eigenvalues of the observable $O$.  Accordingly, for a scalar target variable $y \in \mathbb{R}$, the deterministic predictor is given by the expectation \begin{align}
   \label{eq:deterministic_output_pqc}
    \hat{y} = \text{Tr}(O \rho(x|\theta))=\langle O  \rangle_{\rho(x|\theta)}.
\end{align}Importantly, evaluating the expectation in \eqref{eq:deterministic_output_pqc} to a desired level of accuracy requires carrying out a  large number of measurements. 


Training of the deterministic parametric predictor  \eqref{eq:deterministic_output_pqc} leverages a training set $\mathcal{D}^\text{tr} = \{ z^\text{tr}[i]=(x^\text{tr}[i],y^\text{tr}[i]) \}_{i=1}^{|\mathcal{D}^\text{tr}|}$, and it produces an optimized parameter vector $\theta_{\mathcal{D}^\mathrm{tr}}$.

\begin{figure*}
  \centering
  \includegraphics[width=0.88\textwidth]{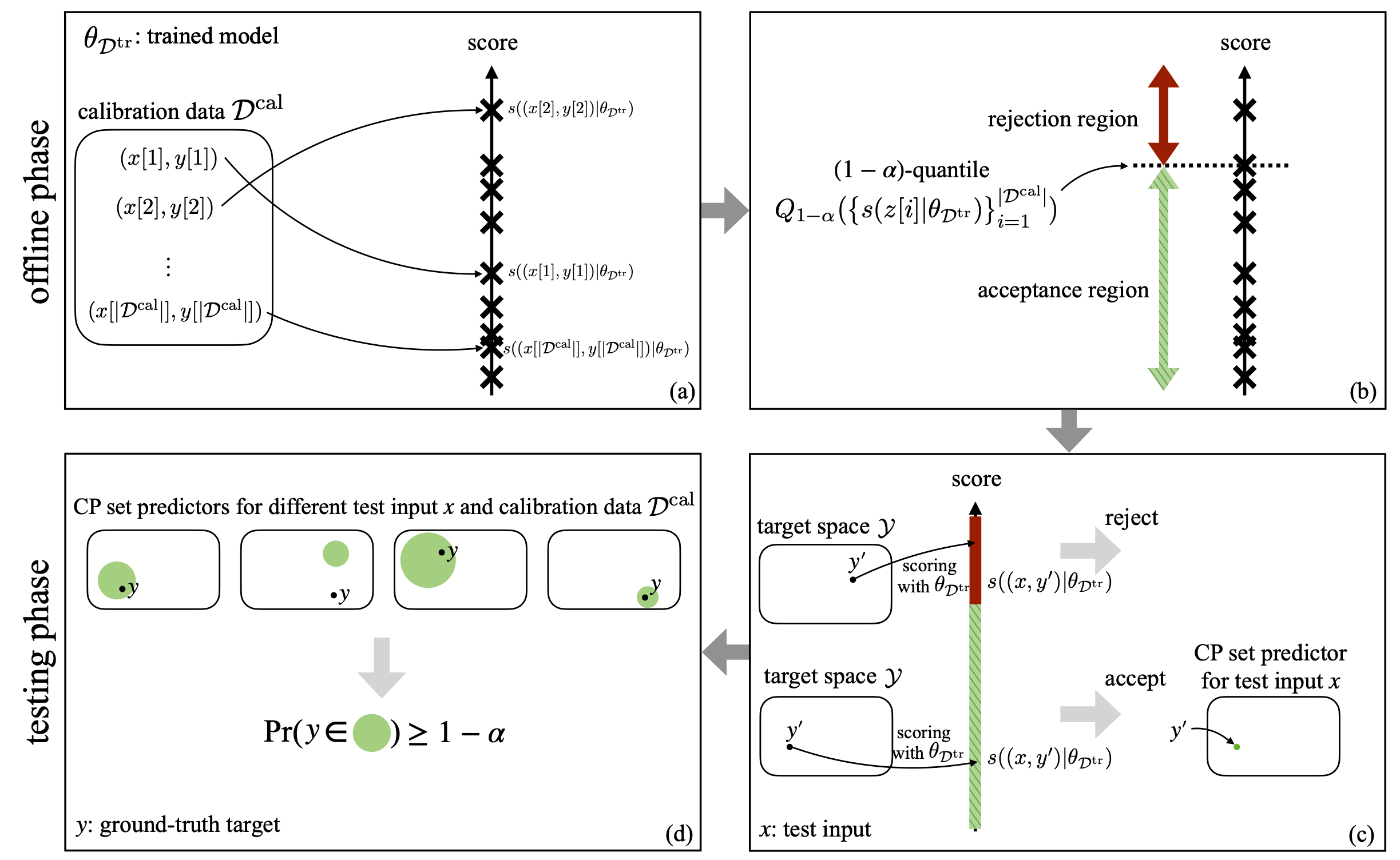}
  \caption{Illustration of conventional CP for deterministic (likelihood-based) models: (a) Based on any trained model parameter vector  $\theta_{\mathcal{D}^\text{tr}}$ and given a scoring function $s(\cdot|\theta_{\mathcal{D}^\text{tr}})$, e.g., the quadratic loss, CP first computes scores for every calibration example in an offline phase; (b) Based on the  $|\mathcal{D}^\text{cal}|$ calibration scores, CP divides the real line into an acceptance region and a rejection region by computing the $(1-\alpha)$-quantile of the calibration scores; (c) CP constructs a set predictor for a test input $x$ by collecting all  candidate outputs $y' \in \mathcal{Y}$ whose scores lie in the acceptance region; (d) The obtained CP set predictor (green circle) is guaranteed to satisfy the validity condition \eqref{eq:relfirst}.}
  \label{fig:conven_CP_cartoon}
\end{figure*}

\subsection{Conformal Prediction for Deterministic PQC Models} \label{subsec:CP_for_infinite_shot}
Assume now access to a pre-trained model and to a calibration data set $\mathcal{D}^\text{cal} = \{ z[i]=(x[i],y[i]) \}_{i=1}^{|\mathcal{D}^\text{cal}|}$.  Following the CP methodology, we are interested in this section in designing a \emph{set predictor} $\Gamma(x|\mathcal{D}^\text{cal}, \theta_{\mathcal{D}^\text{tr}}) \subseteq \mathcal{Y}$ that satisfies the \emph{calibration} condition \begin{equation}\label{eq:relfirst} \Pr\big(\mathbf{y} \in {\Gamma}(\mathbf{x}|\mathbfcal{D}^\text{cal}, \theta_{\mathcal{D}^\text{tr}}) \big) \geq 1-\alpha,\end{equation} for some predetermined \emph{coverage level} $1-\alpha \in (0,1)$. In (\ref{eq:relfirst}),  the probability $\Pr(\cdot)$ is taken over the joint distribution $p(\mathcal{D}^\text{cal},z)$ of calibration data set $\mathcal{D}^\text{cal}$ and test data point $z=(x,y)$.  Condition \eqref{eq:relfirst} stipulates that the predictive set ${\Gamma}(\mathbf{x}|\mathbfcal{D}^\text{cal}, \theta_{\mathcal{D}^\text{tr}})$ must include the true output $\mathbf{y}$ with probability no smaller than $1-\alpha$.

Conventional CP can be directly applied for this purpose. To this end, let us fix  a scoring function $s(z=(x,y)|\theta_{\mathcal{D}^\text{tr}})$ obtained from a loss function such as the quadratic loss $s(z|\theta_{\mathcal{D}^\text{tr}})=(y-\langle O  \rangle_{\rho(x|\theta_{\mathcal{D}^\mathrm{tr}})})^2$, and assume that the miscoverage level satisfies the inequality $\alpha \geq 1/(|\mathcal{D}^\text{cal}|+1)$. During an \emph{offline} calibration phase, CP computes the $\lceil (1-\alpha)(|\mathcal{D}^\text{cal}|+1) \rceil$-th smallest value among the scores $\{ s(z[1]|\theta_{\mathcal{D}^\text{tr}}), ..., s(z[|\mathcal{D}^\text{cal}|]|\theta_{\mathcal{D}^\text{tr}}) \}$ evaluated on the calibration data set, which we denote as $Q_{1-\alpha}( \{  s(z[i]| \theta_{\mathcal{D}^\text{tr}}  )  \}_{i=1}^{|\mathcal{D}^\text{cal}|} )$.

During the \emph{test} phase, CP constructs a predictive set for test input $x$ by including all output values $y\in\mathcal{Y}$ whose score is no larger than $Q_{1-\alpha}( \{  s(z[i]| \theta_{\mathcal{D}^\text{tr}} )  \}_{i=1}^{|\mathcal{D}^\text{cal}|} )$, i.e., 
\begin{align}
    \label{eq:classical_CP_set_predictor}
    \Gamma(x|\mathcal{D}^\text{cal},\theta_{\mathcal{D}^\text{tr}}) =&\Big\{ y' \in \mathcal{Y} \hspace{-0.1cm}\: : \hspace{-0.1cm}\: s((x,y')|\theta_{\mathcal{D}^\text{tr}})  \:\leq\:  Q_{1-\alpha} \big(\big\{ s({z}[i]|\theta_{\mathcal{D}^\text{tr}})\big\}_{i=1}^{ |\mathcal{D}^\text{cal}| } \big) \Big\}.
\end{align}
The overall procedure of CP is illustrated in Fig.~\ref{fig:conven_CP_cartoon}. It can be proved that this predictor satisfies the reliability condition \eqref{eq:relfirst}  as long as the calibration and test data are independent and identically distributed (i.i.d.) (for extensions, see \cite{vovk2022algorithmic, angelopoulos2021gentle, cp2023}).



%% file: Sections/4_QCP_AISTATS.tex
\section{Quantum Circuits as Implicit Probabilistic Models}
\label{sec:PQC_as_implicit}
In this section, we provide the necessary background on PQCs used as probabilistic models by first assuming ideal, noiseless, quantum circuits, and then covering the impact of quantum hardware noise.

\subsection{Noiseless Quantum Circuits as Implicit Probabilistic Models} 
\label{subsec:pqcs_as_implicit_prob_model}
As discussed in Sec.~\ref{subsec:quantum_models}, the output of PQCs is typically taken to be the expectation \eqref{eq:deterministic_output_pqc}. However, the exact evaluation of this quantity requires running the PQC for an arbitrarily large number of times in order to average over a, theoretically infinite, number of shots. As in Sec.~\ref{sec:pqc_deterministic}, let us fix a trained model parameter vector $\theta_{\mathcal{D}^\text{tr}}$, as well as a projective measurement defined by the projection matrices $\{\Pi_j\}_{j=1}^{N'}$ and corresponding numerical outputs. For any input $x$, by Born's rule, the output $\hat{\mathbf{y}}$ obtained from the model equals value $o_j$ with probability \begin{align}
    \label{eq:single_measurement_distribution}
     p(y=o_j|x,\theta_{\mathcal{D}^\text{tr}})=\text{Tr}(\Pi_j \rho(x|\theta_{\mathcal{D}^\text{tr}})).
\end{align} 
The distribution  \eqref{eq:single_measurement_distribution} is generally not directly accessible. Rather, the model is \emph{implicit}, in the sense that all that can be observed by a user are samples $\hat{\mathbf{y}}\sim  p(y|x,\theta_{\mathcal{D}^\text{tr}})$ drawn from it.

More precisely, in an ideal implementation, each new $m$-th measurement shot for a given input $x$ produces an independent output $\hat{\mathbf{y}}^{m}$ with distribution \eqref{eq:single_measurement_distribution}. Accordingly, given an input $x$, one obtains $M$ independent measurements $\hat{\mathbf{y}}^{1:M} = \{ \hat{\mathbf{y}}^m \}_{m=1}^M$ as 
\begin{align} \label{eq:noiseless_iid_pqc}
\hat{\mathbf{y}}^m\underset{\text{i.i.d.}}{\sim} p(y|x,\theta_{\mathcal{D}^\text{tr}}).  
\end{align}

\subsection{Noisy Quantum Circuits as Implicit Probabilistic Models}
\label{subsec:noisy_quantum}

Treating a PQC as an implicit probabilistic model has the additional advantage that one can seamlessly account for the presence of \emph{quantum hardware noise}. As summarized in reference \cite{georgopoulos2021modeling}, there are different sources of quantum noise, with the most relevant for our study being \emph{gate noise}, affecting the internal operations of the quantum circuit, and \emph{measurement noise}, impairing the output measurements.

\emph{Gate noise} can be modelled by appending quantum channels, such as depolarizing and amplitude damping channels, to all, or some, of the gates in the circuit. Overall, the presence of gate noise produces a modified density matrix $\tilde{\rho}(x|\theta_{\mathcal{D}^\text{tr}})$ in lieu of the noiseless density $\rho(x|\theta_{\mathcal{D}^\text{tr}})$. For instance, using Pauli channels to model noise, the noisy density $\tilde{\rho}(x|\theta_{\mathcal{D}^\text{tr}})$ can be expressed as 
\begin{align} \label{eq:gate_noise}
    \tilde{\rho}(x|\theta_{\mathcal{D}^\text{tr}}) = (1-\gamma)\cdot \rho(x|\theta_{\mathcal{D}^\text{tr}}) + \gamma \cdot \rho^{\mathcal{N}}(x|\theta_{\mathcal{D}^\text{tr}}),
\end{align}
where $\rho^{\mathcal{N}}(x|\theta_{\mathcal{D}^\text{tr}})$ represents the \emph{noise density matrix} and $\gamma \in [0,1]$ is a parameter \cite{koczor2021dominant}. As an example, the noise density matrix may be the fully mixed state \cite{schwarz2011detecting}. Note that parameter $\gamma$ determines the level of noise, with $\gamma=0$ indicating a noiseless circuit and $\gamma=1$ a completely noisy circuit whose output state does not depend on the input $x$.

The noisy PQC can be also viewed as an implicit probabilistic model in that, for any input $x$,  it produces a sample $\tilde{\mathbf{y}}$ equal to $o_j$ with probability\begin{align}
    \label{eq:single_measurement_distribution_noisy}
     p(\tilde{y}=o_j|x,\theta_{\mathcal{D}^\text{tr}})=\text{Tr}(\Pi_j \tilde{\rho}(x|\theta_{\mathcal{D}^\text{tr}}))
\end{align} for all $j=1,...,N'$. 
The bias in the expectation \eqref{eq:single_measurement_distribution_noisy} caused by quantum noise can be mitigated via classical post-processing by various \emph{QEM techniques}, while generally increasing the variance of the output \cite{cai2022quantum, jose2022error}.


\setlength{\abovedisplayskip}{3pt}
\setlength{\belowdisplayskip}{0pt}

\emph{Measurement noise} can be modelled as a classical channel that affects the measurement outputs \cite{bravyi2021mitigating, alexander2020qiskit}. Given the output $\tilde{\mathbf{y}} \sim p(y|x, \theta_{\mathcal{D}^\text{tr}})$ from \eqref{eq:single_measurement_distribution_noisy}, measurement noise causes the recorded value $\hat{\mathbf{y}} = \hat{y}$ to be the noisy version 
\begin{align} \label{eq:measurement_noise_single}
    \hat{\mathbf{y}} \sim p(\hat{y}|\tilde{{y}}),
\end{align}
with some conditional distribution $p(\hat{y}|\tilde{y})$. Marginalizing out the noise-free measurement output $\tilde{\mathbf{y}}$ yields the effective distribution of the output $\hat{\mathbf{y}}$ given the input $x$ as  \cite{jayakumar2023universal}
\begin{align}
p(\hat{{y}}=o_i|x,\theta_{\mathcal{D}^\text{tr}}) &= \sum_{j=1}^{N'} p(\hat{y}=o_i|\tilde{y}=o_j) p(\tilde{y}=o_j|x,\theta_{\mathcal{D}^\text{tr}})  \nonumber\\ &= \sum_{j=1}^{N'} p(\hat{y}=o_i|\tilde{y}=o_j)   \text{Tr}(\Pi_j \tilde{\rho}(x|\theta_{\mathcal{D}^\text{tr}})).
\label{eq:single_measurement_distribution_noisy_full}     
\end{align}

\setlength{\abovedisplayskip}{8pt}
\setlength{\belowdisplayskip}{8pt}

Overall, under both gate and measurement noise, given an input $x$, the PQC produces samples
\begin{align} \label{eq:noisy_shot_iid}
\hat{\mathbf{y}}^m\underset{\text{i.i.d.}}{\sim} p(\hat{y}|x,\theta_{\mathcal{D}^\text{tr}})
\end{align}
for $m=1,...,M$, with distribution \eqref{eq:single_measurement_distribution_noisy_full}. QEM for measurement noise aims at inverting  the classical channel $p(\hat{y}|\tilde{y})$ based on an estimate of the channel, which may increase the measurement variance \cite{nation2021scalable, bravyi2021mitigating, georgopoulos2021modeling}.

\subsection{Modelling Quantum Noise Drift}
\label{subsec:noisy_quantum_corr}

So far, by \eqref{eq:noisy_shot_iid}, we have assumed the independence and statistical equivalence of the measurement shots for any given input $x$.  However, in general, quantum noise may exhibit correlations \cite{sarovar2020detecting, bravyi2021mitigating} and drifts \cite{schwarz2011detecting, van2013quantum},  which break the i.i.d. assumption across measurement shots. In the presence of noise correlation and/or drift across shots, one needs to consider the joint distribution $p(\hat{y}^{1:M}|x, \theta_{\mathcal{D}^\text{tr}})$ of all shots $\hat{\mathbf{y}}^{1:M}$ produced for a given input $x$, i.e., $\hat{\mathbf{y}}^{1:M} \sim p(\hat{y}^{1:M}|x,\theta_{\mathcal{D}^\text{tr}})$. 

This more general setting can be modelled in a manner similar to \eqref{eq:single_measurement_distribution_noisy_full}. In particular, the joint distribution $p(\hat{y}^{1:M}|x, \theta_{\mathcal{D}^\text{tr}})$ can be written as
\begin{align} \label{eq:joint_disribution_M_shots}
    p(\hat{y}^{1:M}|x,\theta_{\mathcal{D}^\text{tr}}) = \sum_{\tilde{y}^{1:M}} p(\hat{y}^{1:M}|\tilde{y}^{1:M}) p(\tilde{y}^{1:M}|x,\theta_{\mathcal{D}^\text{tr}}),
\end{align}
where the conditional distribution $p(\hat{y}^{1:M}|\tilde{y}^{1:M})$ represents the classical channel modelling measurement noise, and $p(\tilde{y}^{1:M}|x,\theta_{\mathcal{D}^\text{tr}})$ describes the joint distribution of the observations accounting also for gate noise. The classical channel $p(\hat{y}^{1:M}|\tilde{y}^{1:M})$ may be modelled by using Markov models that capture correlations across shots of the measurement noise \cite{bravyi2021mitigating}. 


A possible model for the distribution $p(\tilde{y}^{1:M}|x,\theta_{\mathcal{D}^\text{tr}})$ was introduced in \cite{schwarz2011detecting, van2013quantum} to study drifts of gate noise in quantum circuits. According to this model, the joint distribution is written as 
\begin{align}\label{eq:drift_joint_distribution}
    p(\tilde{y}^{1:M}|x,\theta_{\mathcal{D}^\text{tr}}) = \prod_{m=1}^M p(\tilde{y}^m|x,\theta_{\mathcal{D}^\text{tr}}), 
\end{align}
with each probability in \eqref{eq:drift_joint_distribution} dependent on a different density matrix $\tilde{\rho}_m(x|\theta_{\mathcal{D}^\text{tr}})$ as
\begin{align}     p(\tilde{{y}}^m=o_j|x,\theta_{\mathcal{D}^\text{tr}})=\text{Tr}(\Pi_j \tilde{\rho}_m(x|\theta_{\mathcal{D}^\text{tr}})).
\end{align}
Assuming Pauli noise channels, the $m$-th density matrix may be expressed by generalizing the expression \eqref{eq:gate_noise} in which the parameter $\gamma$ now depends on the shot index $m$, i.e., 
\begin{align}  \label{eq:decoherence_evolution} 
    \tilde{\rho}_m(x|\theta_{\mathcal{D}^\text{tr}}) = (1-\gamma_m)\cdot \rho(x|\theta_{\mathcal{D}^\text{tr}}) + \gamma_m \cdot \rho^{\mathcal{N}}(x|\theta_{\mathcal{D}^\text{tr}}),
\end{align}
where $\gamma_m \in [0,1]$. Reference \cite{schwarz2011detecting} defined the $m$-th parameter $\gamma_m$ as a function of a parameter $\tau > 0$ as $\gamma_m = 1-e^{-m/\tau}$.

\begin{figure*}
\centering
\includegraphics[width=0.88\textwidth]{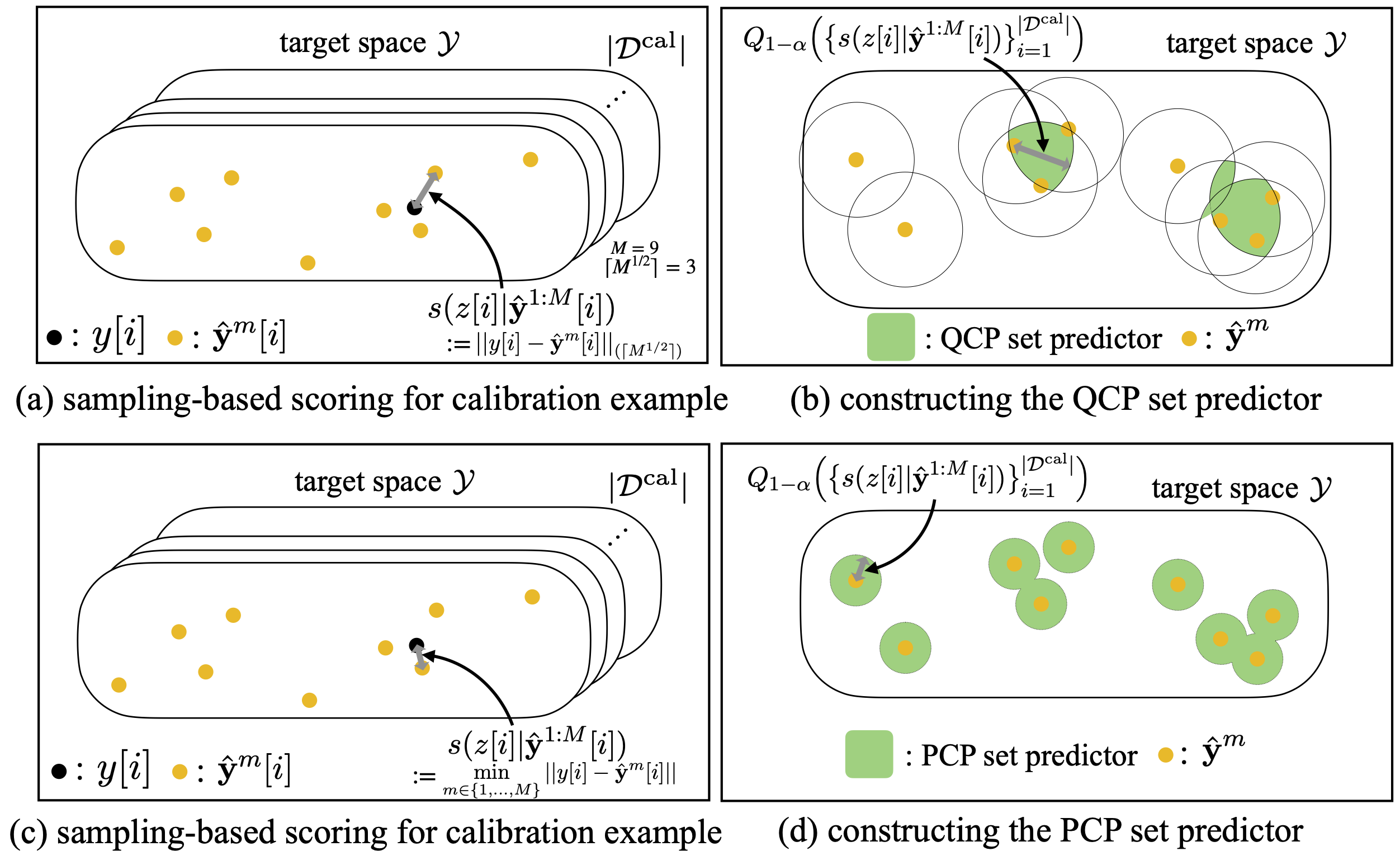}
\caption{(a) In QCP, the score for each calibration example $z[i]=(x[i],y[i])$ is evaluated based on independent random predictions $\hat{\mathbf{y}}^{1:M}[i]$ produced by the trained model as in  \eqref{eq:joint_disribution_M_shots}. Using the proposed scoring function \eqref{eq:pcp_scoring_from_generalized},  the score is obtained by evaluating the $k$-th smallest distance between the true output $y[i]$ and the $M$ random predictions $\hat{\mathbf{y}}^m[i]$.  (b) The QCP set predictor uses the $M$ random predictions $\hat{\mathbf{y}}^{1:M}$ for test input $x$ in order to evaluate the score of every possible candidate output $y'\in\mathcal{Y}$. QCP yields a predictive set given by disjoint regions in the plane. By not relying solely on the closest prediction as PCP \cite{wang2022probabilistic}, as depicted in (c) and (d), the prediction sets produced by QCP can be more robust to noisy predictions (see Sec.~\ref{sec:experimental_results}). This set is guaranteed to contain the true output $\mathbf{y}$ with the predetermined coverage level $1-\alpha$.}
  \label{fig:qcp_scoring} \vspace{-0.3cm}
\end{figure*}

\section{Quantum Conformal Prediction} \label{sec:qcp}
\label{sec:QCP}
In this section, we introduce QCP, a post-hoc calibration scheme for QML that treats PQCs as implicit, i.e., likelihood-free, probabilistic models. We focus on the general model \eqref{eq:joint_disribution_M_shots}, which accounts for quantum gate noise and measurement noise, including also correlation and drifts (see Sec.~\ref{sec:PQC_as_implicit}).

\subsection{Scoring Functions via Non-Parametric Density Estimation} 
\label{subsubsec:quantum_scoring}
The main idea underlying QCP is to build predictive intervals based on a new class of scoring functions that can account for: (\emph{i}) the availability of multiple random predictions, or shots; (\emph{ii}) the inherent noise caused by imperfect models affected by quantum gate noise and measurement noise; and (\emph{iii}) the possible correlation and drifts of the measurement shots. The proposed class of scoring functions builds on \emph{non-parametric density estimation} (see, e.g., \cite{inzenman1991recent, lei2014distribution}), and it recovers the PCP scoring function \cite{wang2022probabilistic} as a special case. 

Given $M$ samples $\hat{\mathbf{y}}^{1:M}$ obtained from \eqref{eq:joint_disribution_M_shots}, we consider scoring functions of the form 
\begin{align} \label{eq:new_class_scoring}
    s(z=(x,y)|\hat{\mathbf{y}}^{1:M}) = \frac{1}{\hat{p}(y|x,\theta_{\mathcal{D}^\text{tr}})},
\end{align} where $\hat{p}(y|x,\theta_{\mathcal{D}^\text{tr}})$ is a non-parametric estimate of the predictive distribution $p(y|x,\theta_{\mathcal{D}^\text{tr}})$ obtained using the samples $\hat{\mathbf{y}}^{1:M}$. Specifically, we focus on non-parametric density estimators of the general form
\begin{align} \label{eq:general_non_para_DE}
    \hat{p}(y|x,\theta_{\mathcal{D}^\text{tr}}) = \frac{1}{M}  \sum_{m=1}^M w_m \cdot  K\bigg( \frac{y - \hat{\mathbf{y}}^m}{h(y,\hat{\mathbf{y}}^{1:M})} \bigg),
\end{align}
where $\{w_m\}_{m=1}^M$ are non-negative weights, $K(\cdot)$ is a \emph{kernel function}, and $h(\cdot,\cdot)$ is a \emph{bandwidth function}. Conventional estimators set the weights $w_m$ to $1$ \cite{cacoullos1964estimation, hand1982kernel}, but, as we will see, more general choices for the weights can better account for the presence of noise drift. Estimators of the form \eqref{eq:general_non_para_DE} recover standard kernel density estimation by assuming a constant bandwidth function; the standard histogram for discrete target space $\mathcal{Y}$ by choosing the kernel $K(x)=\mathbbm{1}(x=0)$, where $\mathbbm{1}(\cdot)$ is the indicator function; and  $k$-nearest neighbor ($k$-NN) density estimators \cite{loftsgaarden1965nonparametric}. 

$k$-NN estimators are obtained by setting the bandwidth function as $h(y, \hat{\mathbf{y}}^{1:M}) = || y - \hat{\mathbf{y}}^m||_{(k)}$, where $|| y - \hat{\mathbf{y}}^m||_{(k)}$ is the $k$-th smallest Euclidean distance in the set of distances $\{ ||y - \hat{\mathbf{y}}^m || \}_{m=1}^M$, and by choosing a rectangular kernel. Accordingly, with the $k$-NN estimator, neglecting constants, we can write the scoring function \eqref{eq:new_class_scoring} as 
\begin{align} 
    \label{eq:pcp_scoring_from_generalized_with_k}
    s(z=(x,y)|\hat{\mathbf{y}}^{1:M}) = || y-\hat{\mathbf{y}}^m||_{(k)}.  
\end{align} 
The $k$-NN density estimation has the property of \emph{asymptotic consistency}, i.e., it tends to the true distribution $p(y|x,\theta)$  pointwise  in probability as long as the parameter $k$ is chosen to grow with $M$, so that the conditions $\lim_{M \rightarrow \infty}k(M)=\infty$ and $\lim_{M \rightarrow \infty}k(M)/M=0$ are verified \cite{loftsgaarden1965nonparametric}.

The scoring function introduced by PCP \cite{wang2022probabilistic} can be obtained as a special case of the $k$-NN scoring function \eqref{eq:pcp_scoring_from_generalized} by setting $k=1$. As mentioned, the choice $k=1$ does not assure asymptotic consistency \cite{loftsgaarden1965nonparametric}. Furthermore, it tends to be sensitive to noise, as it relies solely on the closest sample  as illustrated in Fig.~\ref{fig:qcp_scoring}.

\subsection{Quantum Conformal Prediction} 
\label{subsec:qcp}
QCP adopts the  scoring functions \eqref{eq:new_class_scoring} with the following specific choices: (\emph{i}) for classification problems, we adopt the histogram-based score 
\begin{align} \label{eq:qcp_scoring_classification}
    s(z=(x,y)|\hat{\mathbf{y}}^{1:M}) = \frac{M}{\sum_{m=1}^M w_m \cdot \mathbbm{1}(y=\hat{\mathbf{y}}^m)};
\end{align}
while (\emph{ii}) for regression problems, we adopt the $k$-NN score \eqref{eq:pcp_scoring_from_generalized_with_k}
\begin{align}
    \label{eq:pcp_scoring_from_generalized}
    s(z=(x,y)|\hat{\mathbf{y}}^{1:M}) = || y-\hat{\mathbf{y}}^m||_{( \lceil M^{1/2} \rceil )},
\end{align} 
where the choice of the $k=\lceil M^{1/2} \rceil$ in \eqref{eq:pcp_scoring_from_generalized} is motivated by asymptotic consistency \cite{loftsgaarden1965nonparametric}.

\emph{Offline phase: } During the offline calibration phase, QCP produces $M$ predictions $\hat{\mathbf{y}}^{1:M}[i]$ for each of the calibration examples $z[i] \in \mathcal{D}^\text{cal}$ by using the pre-trained PQC.  Predictions are distributed as in \eqref{eq:joint_disribution_M_shots}, accounting for quantum gate noise and measurement noise, including also correlation and drifts.  After computing the scores for every calibration example $z[i]$ based on the corresponding $M$ predictions $\hat{\mathbf{y}}^{1:M}[i]$, QCP evaluates the $\lceil (1-\alpha)(|\mathcal{D}^\text{cal}|+1)\rceil$-th smallest score (for $\alpha \geq 1/(|\mathcal{D}^\text{cal}|+1)$), which is denoted as $Q_{1-\alpha}( \{  s(z[i]| \hat{\mathbf{y}}^{1:M}[i]  )  \}_{i=1}^{|\mathcal{D}^\text{cal}|} )$ as in Sec. 3.

\emph{Test phase: }  For a test input $x$, QCP produces  $M$ predictions $\hat{\mathbf{y}}^{1:M}$, which follow \eqref{eq:joint_disribution_M_shots}, via the same pre-trained PQC. These predictions are used to construct the predicted  set  as
\begin{align}
    \label{eq:QCP_set_predictor}
    \mathbf{\Gamma}_M(x|\mathcal{D}^\text{cal},\theta_{\mathcal{D}^\text{tr}}) =&\Big\{ y' \in \mathcal{Y} \hspace{-0.1cm}\: : \hspace{-0.1cm}\: s((x,y')|\hat{\mathbf{y}}^{1:M} )  \:\leq\:  Q_{1-\alpha} \big(\big\{ s({z}[i]| \hat{\mathbf{y}}^{1:M}[i] )\big\}_{i=1}^{ |\mathcal{D}^\text{cal}| } \big) \Big\}.
\end{align}

We now provide explicit expression for the QCP set predictor for both regression and classification problem. 

\emph{QCP for regression:} Using \eqref{eq:pcp_scoring_from_generalized}, the QCP set predictor \eqref{eq:QCP_set_predictor} is given as
\begin{align}
\label{eq:closed_form_genreal_scoring_pcp_and_qcp}
&\mathbf{\Gamma}_M(x|\mathcal{D}^\text{cal},\theta_{\mathcal{D}^\text{tr}}) =\Big\{ y' \in \mathcal{Y} \hspace{-0.1cm}\: : \hspace{-0.1cm}\: \sum_{m=1}^M \mathbbm{1}\big(y' \in B_{Q(\mathcal{D}^\text{cal})}(\hat{\mathbf{y}}^m)\big) \geq \lceil M^{1/2} \rceil \Big\},
 \end{align}
in which $B_{q(\mathcal{D}^\text{cal})}(y)$ is the closed interval with center point $y$, i.e., $B_r(y) = \{ y' \in \mathcal{Y}: || y - y'|| \leq r \}$, and radius $Q(\mathcal{D}^\text{cal})=  Q_{1-\alpha} \big(\big\{ 
|| y[i] - \hat{\mathbf{y}}^m[i] ||_{(\lceil M^{1/2} \rceil)}
\big\}_{i=1}^{ |\mathcal{D}^\text{cal}| } \big)$. 

\emph{QCP for classification:} Using the weighted histogram  in \eqref{eq:qcp_scoring_classification}, the QCP set predictor \eqref{eq:QCP_set_predictor} is evaluated as
\begin{align}
\label{eq:qcp_set_predictor_classification}
&\mathbf{\Gamma}_M(x|\mathcal{D}^\text{cal},\theta_{\mathcal{D}^\text{tr}}) =\Big\{ y' \in \mathcal{Y} \hspace{-0.1cm}\: : \hspace{-0.1cm}\: \sum_{m=1}^M w_m \cdot \mathbbm{1}(y'=\hat{\mathbf{y}}^m)  \geq \frac{M}{Q(\mathcal{D}^\text{cal})}\Big\}.
\end{align}
An illustration of QCP is provided in Fig.~\ref{fig:QCP_cartoon}, while an algorithmic description can be found in  Algorithm~\ref{alg:QCP}.
We will discuss specific choice for the weights $\{w_m\}_{m=1}^M$ in Sec.~\ref{sec:experimental_results}.

\begin{figure*}[t]
  \centering
  \includegraphics[width=0.9\textwidth]{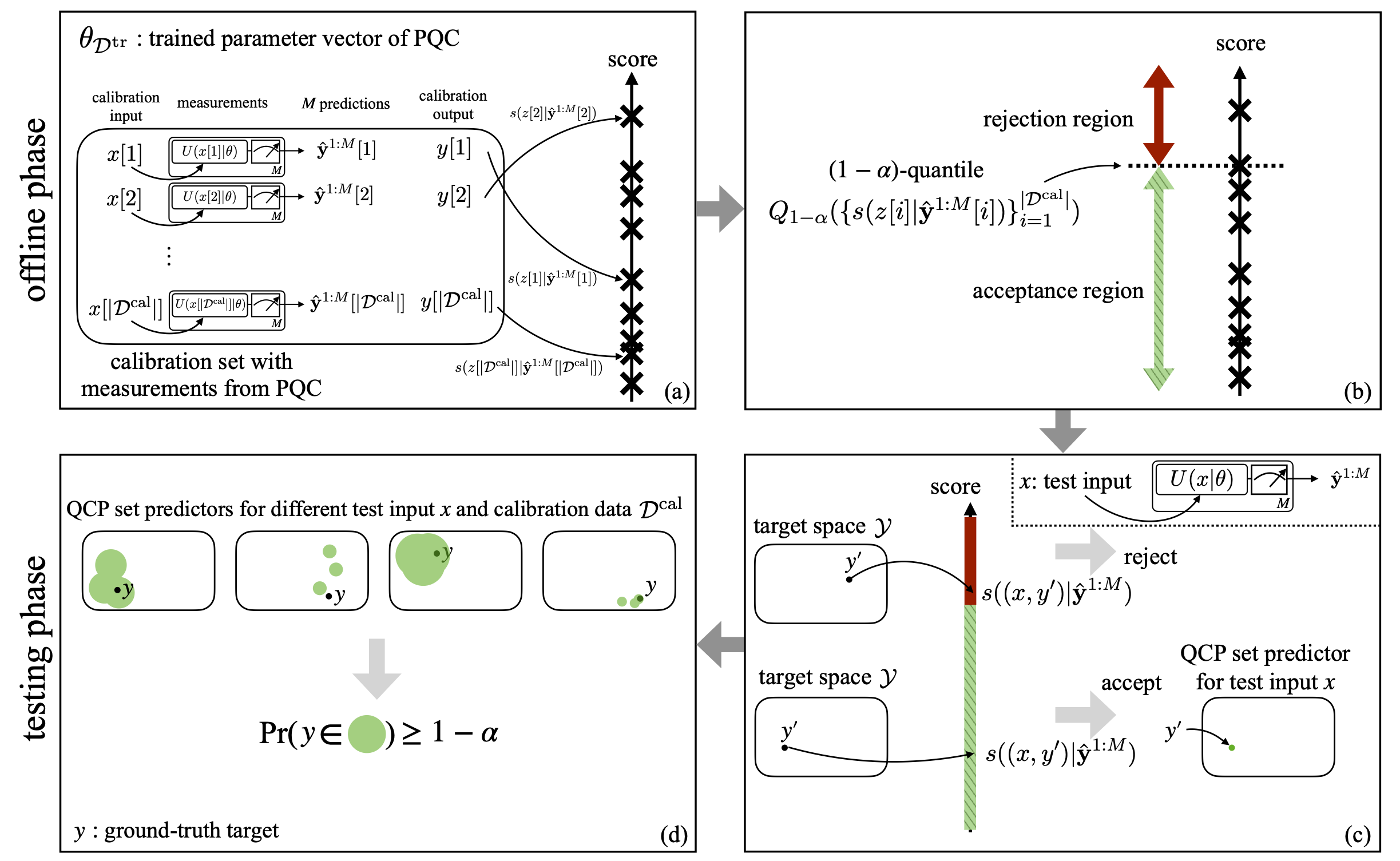}
  \caption{Illustration of QCP for a PQC with trained parameter vector  $\theta_{\mathcal{D}^\text{tr}}$ and for a scoring function $s((x,y)|\hat{\mathbf{y}}^{1:M})$ (see Sec.~4 of the main text): In the \emph{offline phase}, (a) QCP first computes a score for each calibration example based on  $M$ random predictions produced by the PQC; (b) then, based on the  obtained $|\mathcal{D}^\text{cal}|$ calibration scores, QCP divides the real line into an acceptance region and a rejection region by using as a threshold the $(1-\alpha)$-quantile of the calibration scores. In the \emph{testing} phase,  (c) QCP constructs a set predictor for a test input $x$ by including in the set all  candidate outputs $y' \in \mathcal{Y}$ whose scores lie in the acceptance region.  (d) The obtained QCP set predictor (green circle) is guaranteed to satisfy the validity condition \eqref{eq:validity_condition_for_all_theorems}.}
  \label{fig:QCP_cartoon}
\end{figure*}

\begin{algorithm}[hbt!] 
\DontPrintSemicolon
\LinesNumbered
\smallskip
\KwIn{trained parameter vector $\theta_{\mathcal{D}^\text{tr}}$ for the PQC; calibration data set $\mathcal{D}^\text{cal}=\{ z[i]=(x[i],y[i]) \}_{i=1}^{|\mathcal{D}^\text{cal}|}$; test input $x$; number of shots $M\geq 1$; desired coverage level $1-\alpha$}
\KwOut{well-calibrated predictive set $\Gamma_M(x|\mathcal{D}^\text{cal},\theta_{\mathcal{D}^\text{tr}})$ at coverage level $1-\alpha$  }
\vspace{0.15cm}
\hrule
\vspace{0.15cm}
{\bf choose}  scoring function $s( (x,y)|\hat{\mathbf{y}}^{1:M})$ as outlined  in  Sec.~4 of the main text\\
\texttt{Offline phase}\\
\For{{\em $i$-th calibration data example $z[i]$ with $i=1,...,|\mathcal{D}^\text{cal}|$ }}{\vspace{0.1cm}
given input $x[i]$, produce $M$ independent predictions $\hat{\mathbf{y}}^{1:M}[i]=\{ \hat{\mathbf{y}}^{m}[i] \}_{m=1}^{M}$ via $M$ shots of the PQC \\
compute the corresponding score $s(z[i]|\hat{\mathbf{y}}^{1:M}[i])$
}
find the  $\lceil(1-\alpha)(|\mathcal{D}^\text{cal}|+1)\rceil$-th smallest value among the obtained scores, which is denoted as $Q_{1-\alpha}( \{  s(z[i]|\hat{\mathbf{y}}^{1:M}[i])  \}_{i=1}^{|\mathcal{D}^\text{cal}|} )$ \\
\texttt{Testing phase for input $x$}\\
{\bf initialize} predictive set $\mathbf{\Gamma}_M(x|\mathcal{D}^\text{cal},\theta_{\mathcal{D}^\text{tr}}) \leftarrow \{ \}$\\
for test input $x$, make $M$ predictions $\hat{\mathbf{y}}^{1:M}=\{ \hat{\mathbf{y}}^{m} \}_{m=1}^{M}$ via PQC  \\
\For{{\em candidate output $y' \in \mathcal{Y}$}}{\vspace{0.1cm}
compute the corresponding score $s((x,y')|\hat{\mathbf{y}}^{1:M})$\\
\hspace{-0.25cm}\uIf{ ${s}((x,y')|\hat{\mathbf{y}}^{1:M})  \:\leq\:  Q_{1-\alpha}( \{  s(z[i]|\hat{\mathbf{y}}^{1:M}[i])  \}_{i=1}^{|\mathcal{D}^\text{cal}|} )$}{
$\mathbf{\Gamma}_M(x|\mathcal{D}^\text{cal},\theta_{\mathcal{D}^\text{tr}}) \leftarrow \mathbf{\Gamma}_M(x|\mathcal{D}^\text{cal},\theta_{\mathcal{D}^\text{tr}})  \cup \{ y' \}$
  }
}
{\bf return} the predictive set $\mathbf{\Gamma}_M(x|\mathcal{D}^\text{cal},\theta_{\mathcal{D}^\text{tr}})$ for test input $x$
\caption{Quantum Conformal Prediction (QCP)}
\label{alg:QCP}
\end{algorithm}

\subsection{Theoretical Calibration Guarantees}
\label{subsubsec:qcp_theoretical}

The QCP set predictor \eqref{eq:QCP_set_predictor} is well calibrated at coverage level $1-\alpha$, as detailed in the next theorem proved in Appendix~\ref{appendix:proofs}.

\begin{theorem}[Calibration of QCP]\label{ther:qcp} Assuming that calibration data $\mathbfcal{D}^\text{cal}$ and test data $\mathbf{z}$ are i.i.d., while allowing for arbitrary correlation and drifts across shots as per  \eqref{eq:joint_disribution_M_shots},  for any coverage level $1-\alpha \in (0,1)$ and for any pre-trained PQC with model parameter vector $\theta_{\mathcal{D}^{\text{tr}}}$, the QCP set predictor \eqref{eq:QCP_set_predictor} satisfies the inequality \begin{align}
    \label{eq:validity_qcp}
    \Pr( \mathbf{y} \in \mathbf{\Gamma}_M(\mathbf{x}|\mathbfcal{D}^\text{cal},\theta_{\mathcal{D}^\text{tr}}) ) \geq 1-\alpha,
\end{align}with probability taken over the joint distribution of the test data $\mathbf{z}$, of the calibration data set $\mathbfcal{D}^\text{cal}$, and of the independent random predictions $\hat{\mathbf{y}}^{1:M}$ and $\{\hat{\mathbf{y}}^{1:M}[i]\}_{i=1}^{|\mathcal{D}^\text{cal}|}$ produced by the PQC for test and calibration points, respectively.
\end{theorem}

In Appendix~\ref{sec:finite_number_guarantee}, we further elaborate on the practical significance of the probability in \eqref{eq:validity_qcp} by studying the ratio of successful set predictions over a number of trials characterized by independent calibration and test data (see also \cite[Sec.~C]{angelopoulos2021gentle} for conventional CP in classical models).

\section{QCP for Quantum Data} \label{sec:application_to_quantum}
In the previous sections, we have focused on situations in which input data $x$ is of classical nature. Accordingly, given a pre-trained PQC, one can produce $M$ copies of the quantum embedding state $\rho(x|\theta_{\mathcal{D}^\text{tr}})$ to be used by QCP in order to yield the set predictor $\mathbf{\Gamma}_M(x|\mathcal{D}^\text{cal},\theta_{\mathcal{D}^\text{tr}})$. As we briefly discuss here, QCP can be equally well applied to settings in which the pre-trained model takes as input quantum data.

Assume that there are $C$ classes of quantum states, indexed by variable $y \in \{1,...,C\}$, such that the $y$-th class corresponds to a density matrix $\rho(y)$. The data generating mechanism is specified by a distribution $p(y)$ over the class index $y$ and by the set of density matrices $\rho(y)$ with $y\in \{1,...,C\}$. Given any pre-designed positive operator-valued measurement (POVM) defined by positive semidefinite matrices $\mathcal{P}=\{\mathrm{P}_1,...,\mathrm{P}_C\}$,  the POVM produces output $\hat{\mathbf{y}}$ with probability  $
    \hat{\mathbf{y}} \sim \text{Tr}(\mathrm{P}_y \rho ).$
The POVM $\mathcal{P}$ may be implemented via a pre-trained PQC with parameter vector $\theta_{\mathcal{D}^\text{tr}}$. 

The goal of QCP is to use the predesigned model operating on $M$ copies of a test state $\rho \in \{\rho(1),...,\rho(C)\}$, denoted as $\rho^{\otimes M}$, as well as $M$ copies of each example in the calibration data set $\mathcal{D}^\text{cal}$, to produce a set predictor $\mathbf{\Gamma}_M(\rho^{\otimes M}|\mathcal{D}^\text{cal},\mathcal{P})$ that contains the true label with predetermined coverage level $1-\alpha$. The calibration data is of the form $\mathcal{D}^\text{cal}=\{(\rho(y[i])^{\otimes M},y[i])\}_{i=1}^{|\mathcal{D}^\text{cal}|}$, where label  $y[i]\in \{1,...,C\}$ and $M$ copies of corresponding state $\rho(y[i])$ are available for each $i$-th calibration example. Mathematically, the reliability requirement can be expressed as the inequality \begin{align}
        \label{eq:qcp_QC_main}
        \Pr\big(\mathbf{y} \in \mathbf{\Gamma}_M(\rho(\mathbf{y})^{\otimes M}|\mathbfcal{D}^\text{cal},\mathcal{P})\big) \geq 1-\alpha,
    \end{align}
where the probability $\Pr(\cdot)$ is taken over the i.i.d. calibration and test labels $\mathbf{y}, \mathbf{y}[1],...,\mathbf{y}[|\mathcal{D}^\text{cal}|]\sim p(y)$  and over the independent random predictions $\hat{\mathbf{y}}^{1:M}$ and $\{\hat{\mathbf{y}}^{1:M}[i]\}_{i=1}^{|\mathcal{D}^\text{cal}|}$. QCP can be directly applied to the test predictions $\hat{\mathbf{y}}^{1:M}$ and to the calibration predictions $\{\hat{\mathbf{y}}^{1:M}[i]\}_{i=1}^{|\mathcal{D}^\text{cal}|}$ as described in Sec.~\ref{subsec:qcp}, and, by Theorem 3, it satisfies the reliability condition \eqref{eq:validity_qcp}. 

%% file: Sections/6_Experimental_Settings.tex
\section{Experimental Settings}
\label{sec:experimental_settings}
In the rest of this paper, we demonstrate the validity of the proposed QCP method by addressing both an unsupervised learning task, namely density learning, and a supervised learning task, namely regression. We start in this section by describing the experimental settings, including problem definition and assumed PQC ansatzes, while  the next section presents the experimental results. Numerical examples using a classical simulator  were run over a GPU server with single NVIDIA A100 card, while experiments using a NISQ computer were implemented on  the \texttt{imbq\_quito} device made available by IBM Quantum.  We cover first density learning and then regression. 

\begin{figure*}[t]
  \centering 
  \includegraphics[width=0.95\textwidth]{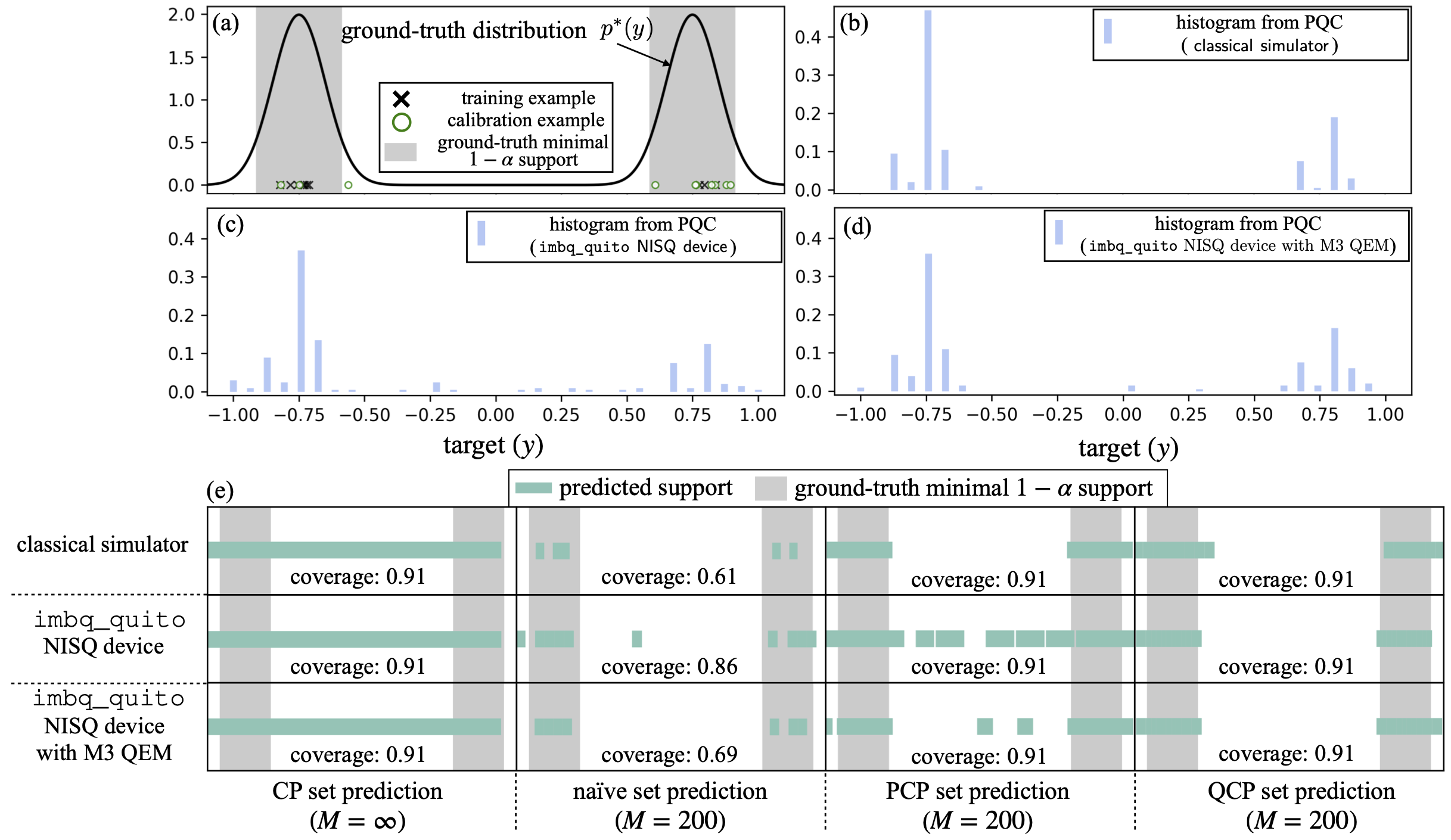}
  \caption{Illustration of the unsupervised learning problem of support estimation, in which the goal is to use the samples drawn from the PQC, as well as additional calibration data $\mathcal{D}^{\text{cal}}$, to estimate the support at coverage level $1-\alpha$ of the  ground-truth probability density $p^*(y)$. (a) Ground-truth (unknown) distribution $p^*(y)$ and smallest support set \eqref{eq:opt_set_predictor_de} with coverage $1-\alpha=0.9$ (gray area), along with $|\mathcal{D}^{\text{tr}}|=10$ training examples (crosses) and $|\mathcal{D}^{\text{cal}}|=10$ calibration examples (circles); Histogram of the $M=200$ samples $\hat{\mathbf{y}}^{1:M}$ obtained via a PQC trained using the data set shown in part (a) and implemented via: (b) a classical simulator; (c) the \texttt{imbq\_quito} NISQ device made available through IBM Quantum; and (d)  the \texttt{imbq\_quito} NISQ device with M3 quantum error mitigation (QEM) \cite{nation2021scalable}. (e) Predicted intervals produced by the na\"ive set predictor \eqref{eq:naive_set_predictor_de}, by {\color{black}conventional CP,  by PCP \cite{wang2022probabilistic}, and by QCP \eqref{eq:closed_form_genreal_scoring_pcp_and_qcp}} for one realization of the calibration data set. The estimated empirical coverage probability for the set predictors is also indicated.}

  \label{fig:density_learning_intro}
\end{figure*}

\subsection{Unsupervised Learning: Density Learning}
\label{subsec:unsup_learning}
Given a data set $\mathcal{D}=\{y[i]\}_{i=1}^{|\mathcal{D}|}$, with training samples $y[i]$ following an unknown population distribution $p^*(y)$ on the real line,  density estimation  aims at inferring some properties about the distribution $p^*(y)$ (see, e.g., \cite[Sec.~7.3]{simeone2022machine}).  We make the standard assumption that the data points in set $\mathcal{D}$ are drawn i.i.d. from the population distribution $p^*(y)$. Following \cite{lei2011efficient}, we specifically focus on the problem of reliably identifying a collection of intervals that are guaranteed to contain a test sample $y\sim p^*(y)$ with coverage probability at least $1-\alpha$, as illustrated in Fig.~\ref{fig:density_learning_intro}(a). We recall that, with CP, PCP, and QCP, the coverage probability is evaluated also with respect to the calibration data set. The ground-truth population distribution $p^*(y)$ is the mixture of Gaussians $p^*(y) = \frac{1}{2}\mathcal{N}(-0.75, 0.1^2)+ \frac{1}{2}\mathcal{N}(0.75, 0.1^2)$ as also shown in Fig.~\ref{fig:density_learning_intro}-(\emph{i}).

\subsubsection{Benchmarks}
\label{subsec:benchmarks_and_qcp_density_learning}
For  conventional CP and for QCP, set predictors are obtained from \eqref{eq:classical_CP_set_predictor} and \eqref{eq:closed_form_genreal_scoring_pcp_and_qcp}, respectively by disregarding the input $x$.  PCP set predictor can be obtained in the same way \cite{wang2022optimizing} (see also Appendix~\ref{sec:probabilistic_CP}). Accordingly, we denote the corresponding set predictors as $\Gamma(\mathbfcal{D}^\text{cal},\theta_{\mathcal{D}^\text{tr}})$ and  $\mathbf{\Gamma}_M(\mathbfcal{D}^\text{cal},\theta_{\mathcal{D}^\text{tr}})$, respectively. As benchmarks, we consider an ideal support estimator and a na\"ive support predictor based on a pre-trained PQC.

The \emph{smallest support set with coverage $1-\alpha$} is given by \begin{align}
    \label{eq:opt_set_predictor_de}
    {\Gamma}^\text{opt} = \arg\min_{\Gamma \in 2^{\mathcal{Y}}} |\Gamma| \text{ s.t. } \int_{y \in \Gamma} p^*(y)\mathrm{d}y \geq 1-\alpha.
\end{align}This is depicted in  Fig.~\ref{fig:density_learning_intro}-(\emph{i}) using  a gray shaded area. This set offers an idealized solution, and it cannot be evaluated in practice, since the ground-truth distribution  $p^*(y)$ is not known. 

Consider now a trained PQC that implements an implicit probabilistic model $p(y|\theta_{\mathcal{D}})$, where $\theta_{\mathcal{D}}$ is the parameter vector optimized based on data set $\mathcal{D}$. Replacing  the ground-truth distribution $p^*(y)$ with the trained probabilistic model $p(y|\theta_{\mathcal{D}})$ in \eqref{eq:opt_set_predictor_de}, a  \emph{na\"ive support predictor} can be obtained as 
\begin{align}
    \label{eq:naive_set_predictor_de}
    {\mathbf{\Gamma}}_M^\text{na\"ive}(\theta_{\mathcal{D}}) = \arg\min_{\Gamma \in 2^{\mathcal{Y}}} |\Gamma| \text{ s.t. } \int_{y \in \Gamma} p(y|\theta_{\mathcal{D}})\mathrm{d}y \geq 1-\alpha.
\end{align} To evaluate the set in (\ref{eq:naive_set_predictor_de}), samples  $\hat{\mathbf{y}}^{1:M}$ drawn from the trained model ${p}(y|\theta_{\mathcal{D}})$ are used  to approximate the integral in (\ref{eq:naive_set_predictor_de}) via Monte Carlo integration.

\begin{figure}
  \centering
  \includegraphics[width=12.5cm]{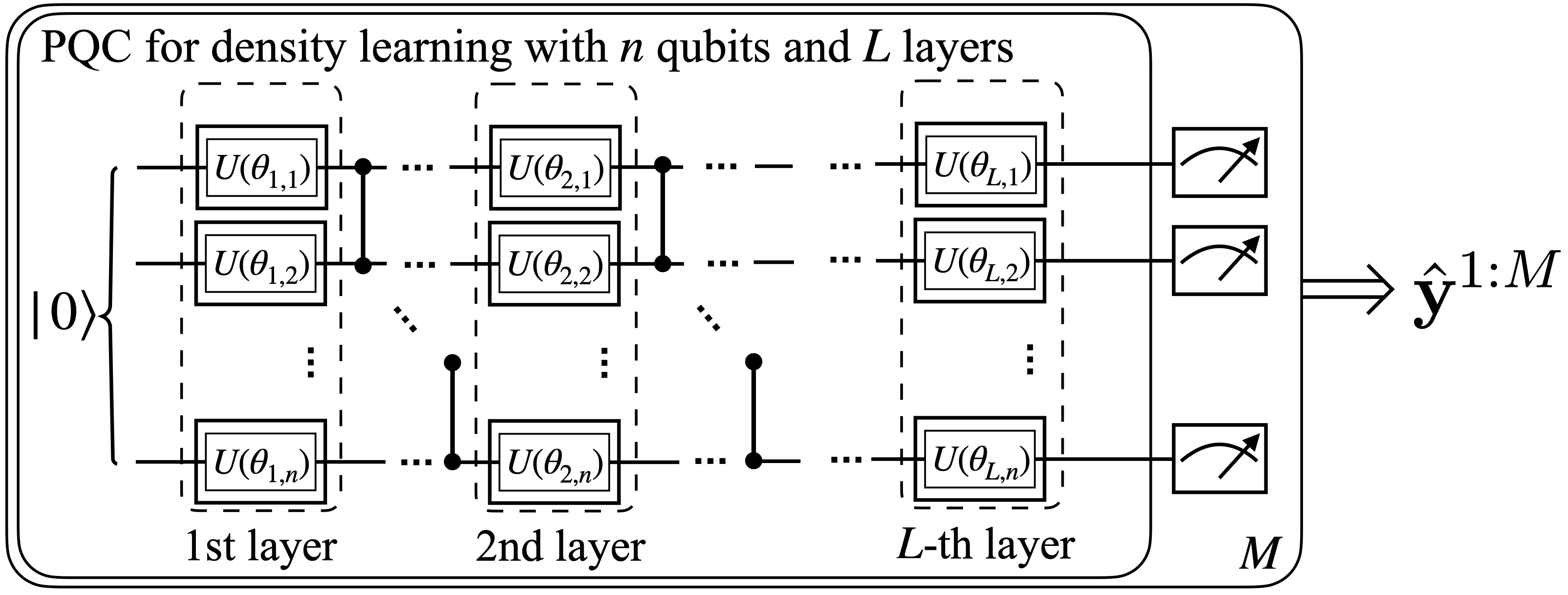}
  \caption{Illustration of the assumed hardware-efficient ansatz for density learning.} 
  \label{fig:pqc_ansatz_density_learning}
\end{figure}
\subsubsection{PQC Ansatz}
\label{subsubsec:pqc_ansatz_without_input}

We adopt the standard \emph{hardware-efficient ansatz}, which has also been previously used for the related unsupervised learning task of training generative quantum models \cite{du2022theory, gili2022evaluating,zoufal2019quantum}.  As shown in Fig.~\ref{fig:pqc_ansatz_density_learning}, the parameterized unitary matrix $U(\theta)$ operates on a register of $n$ qubits, which are initially in the fiducial state $|0\rangle$. The parameterized unitary matrix $U(\theta)$ consists of  $L$ layers applied sequentially as
\begin{align}
    \label{eq:L_layer_unitary_PQC_without_input}
    U(\theta) = U_L(\theta)\cdot U_{L-1}(\theta)\cdots U_1(\theta),
\end{align}
with  $U_l(\theta)$  being the unitary matrix corresponding to the $l$-th layer.

Following the hardware-efficient ansatz, the unitary $U_l(\theta)$ can be written as (see, e.g., \cite[Sec. 6.4.2]{simeone2022introduction})\begin{align}
    \label{eq:parameterized_unitary_for_density_learning}
    U_l(\theta) = U_{ent}(R(\theta_{l,1}^1, \theta_{l,1}^2, \theta_{l,1}^3 ) \otimes \cdots \otimes R(\theta_{l,n}^1, \theta_{l,n}^2, \theta_{l,n}^3 ) ),
\end{align}
where $U_{ent}$ is an \emph{entangling unitary}, while \begin{align}R(\theta^1, \theta^2, \theta^3)=R_Z(\theta^1)R_{Y}(\theta^2)R_{Z}(\theta^3)\end{align} represents a general parameterized single-qubit gate with Pauli rotations $R_{P}(\theta)=\exp(-i \frac{\theta}{2}P)$ for $P\in \{Y,Z\}$. By \eqref{eq:parameterized_unitary_for_density_learning},  the parameter vector $\theta$ is the collection of the angles $[\theta_{l,k}^1, \theta_{l,k}^2, \theta_{l,k}^3]$ for all $k$-th qubits and all $l$-th layers,  i.e., $\theta=\{\{ \theta_{l,k}^1,\theta_{l,k}^2,\theta_{l,k}^3 \}_{k=1}^{n}\}_{l=1}^L$.

\begin{figure*}
  \centering
  \includegraphics[width=0.98\textwidth]{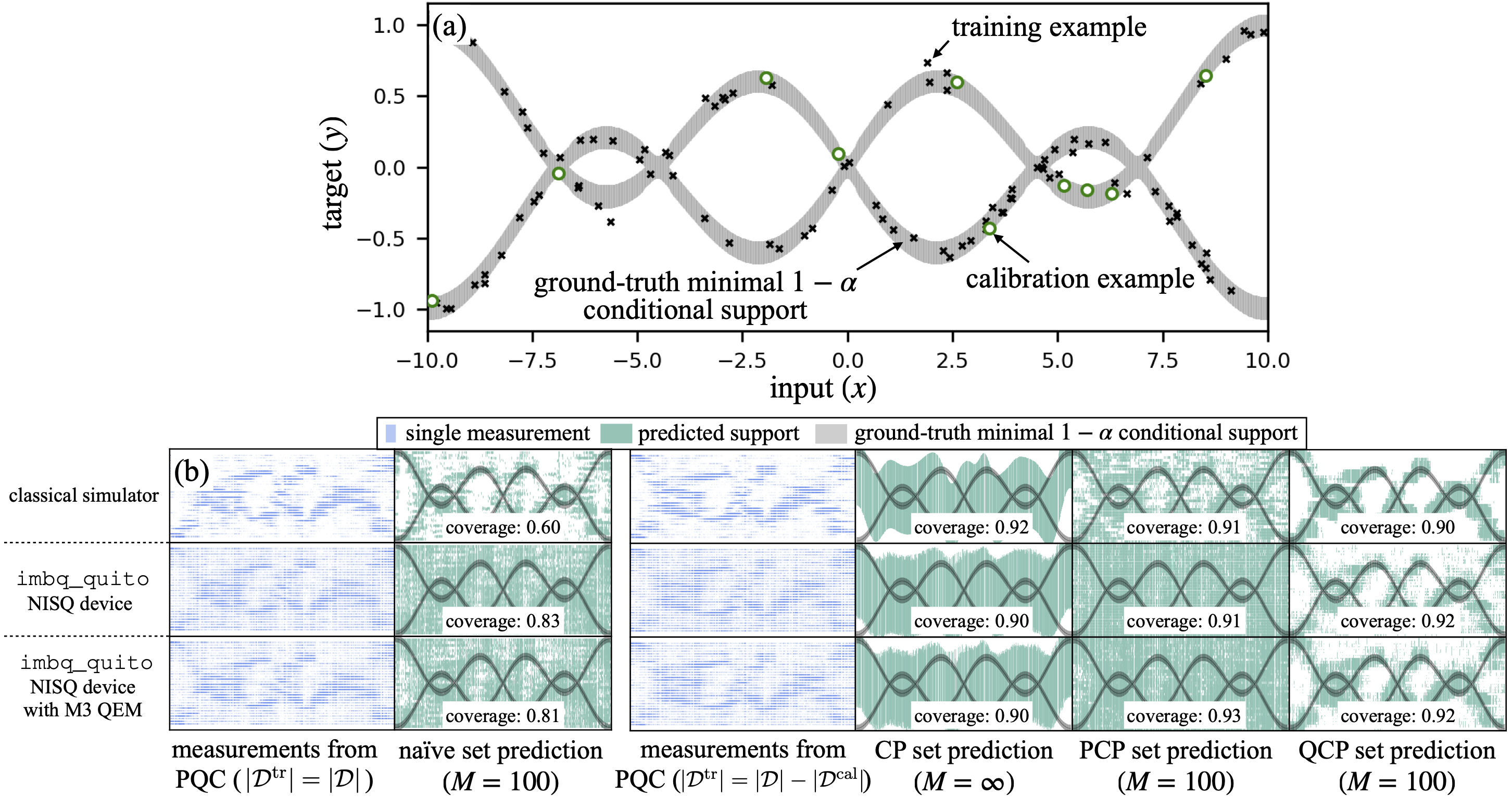}
  \caption{Illustration of the supervised learning problem of regression under study, in which the goal is to use the samples drawn from the PQC, as well as additional calibration data $\mathcal{D}^{\text{cal}}$, in order to produce a subset of predicted values at coverage probability $1-\alpha$ for the target $y$ given test input $x$. (a) Ground-truth (unknown) minimal $1-\alpha$ conditional support \eqref{eq:opt_set_predictor_de_cond} (gray area), along with $|\mathcal{D}^{\text{tr}}|=90$ training examples (crosses) and $|\mathcal{D}^{\text{cal}}|=10$ calibration examples (circles); (b) First and third columns: $M=100$ measurements obtained from a PQC  with trained model  based on the entire data set, $\theta_{\mathcal{D}}$ (first column) and based only on the data partition,  $\theta_{\mathcal{D}^\text{tr}}$  (third column), with the PQC being  implemented either using a classical simulator or the \texttt{imbq\_quito} NISQ device made available through IBM Quantum with or without M3 QEM \cite{nation2021scalable}. (b) Second and fourth columns: Predicted intervals produced by the na\"ive set predictor \eqref{eq:naive_set_predictor_de_cond}, {\color{black} by conventional CP, by PCP \cite{wang2022probabilistic},  and by QCP \eqref{eq:closed_form_genreal_scoring_pcp_and_qcp}} for one realization of the calibration data set. The estimated  empirical coverage probability  via na\"ive set prediction and QCP set prediction  is also indicated.} 
\label{fig:regression_learning_full}
\end{figure*}

Given the PQC described above, for a fixed model parameter vector $\theta$, $M$ predictions $\hat{\mathbf{y}}^{1:M}$ are obtained by measuring a quantum observable $O$ under the quantum state
\begin{align}\rho(\theta)=U(\theta)|0\rangle \langle 0| U(\theta)^\dagger\label{eq:quantum_embedding_without_input}\end{align}  produced as the output of the PQC. Measurements are implemented in the computational basis, i.e., with projection matrices $\Pi_j = |j\rangle \langle j|$, denoting as $|j\rangle$ the one-hot amplitude vector with $1$ at position $j$ for $j=1,...,N$. We choose the possible values  $\{o_j\}_{j=1}^{N}$ to be equally spaced in the interval $[-1,1]$. Accordingly, we set 
the $j$-th eigenvalue of the observable $O$ as $o_{j} = -1 + 2(j-1)/(N-1)$.

We adopt a PQC with $n=5$ qubits, and the entangling unitary consists of  controlled-Z (CZ) gates connecting successive qubits, i.e., $U_{ent} = \prod_{k=1}^{n-1}C^{Z}_{k,k+1}$, where $C^{Z}_{k,k+1}$ is the CZ gate between the $k$-th qubit and $(k+1)$-th qubit. We set number of layers to $L=2$.  We implemented the described PQC on (\emph{i}) a classical simulator; (\emph{ii}) on the \texttt{imbq\_quito} NISQ device made available through IBM Quantum; and (\emph{ii}) on the \texttt{imbq\_quito} NISQ device with M3 quantum error mitigation \cite{nation2021scalable}.


\subsection{Supervised Learning: Regression}
\label{subsec:supervised_learning}
In the regression problem under study, we aim  at predicting a scalar continuous-valued target $y$ given scalar input $x$, with input and output following the unknown ground-truth joint distribution $p^*(x,y)$. We assume access to a data set $\mathcal{D}=\{z[i]=(x[i],y[i])\}_{i=1}^{|\mathcal{D}|}$, with training samples $z[i]\sim p^*(x,y)$ drawn in an i.i.d. manner.

We specifically consider the mixture of two sinusoidal functions for the ground-truth distribution, as given by \begin{align}
    \label{eq:regression_target}
    p^*(y|x) = \frac{1}{2}\mathcal{N}(\mu(x), 0.05^2) + \frac{1}{2}\mathcal{N}(-\mu(x), 0.05^2) \text{ and } p^*(x) = \mathcal{U}(-10,10),
\end{align}
where we have denoted as  $\mathcal{U}(a,b)$ the uniform distribution in the interval $[a,b]$, and we set $\mu(x) = 0.5 \sin(0.8x) + 0.05x$. We note that bimodal distributions  such as (\ref{eq:regression_target}) cannot be represented by standard classical machine learning models used for regression that assume a Gaussian conditional distribution $p(y|x,\theta)$ with a parameterized mean \cite{masegosa2020learning, morningstar2022pacm, zecchin2022robust}. 

\subsubsection{Benchmarks}
Benchmarks for regression generalize the approaches described in Sec.~\ref{subsec:benchmarks_and_qcp_density_learning} for unsupervised learning. Accordingly, the \emph{smallest prediction interval with coverage $1-\alpha$ for input $x$} is given by \begin{align}
    {\Gamma}^\text{opt}(x) = &\arg\min_{\Gamma \in 2^{\mathcal{Y}}} |\Gamma| \text{ s.t. } \int_{y \in \Gamma} p^*(y|x)\mathrm{d}y \geq 1-\alpha.\label{eq:opt_set_predictor_de_cond}
\end{align} This is illustrated in  Fig.~\ref{fig:regression_learning_full}(a) as a gray shaded area. Since this ideal interval cannot be computed, a \emph{na\"ive} \emph{set predictor} alternative is to replace the ground-truth distribution $p^*(y|x)$ with a model $p(y|x,\theta_{\mathcal{D}})$ trained using the available data $\mathcal{D}$. This yields
\begin{align}
    {\mathbf{\Gamma}}_M^\text{na\"ive}(x|\theta_{\mathcal{D}}) = &\arg\min_{\Gamma \in 2^{\mathcal{Y}}} |\Gamma| \text{ s.t. } \int_{y \in \Gamma} p(y|x,\theta_{\mathcal{D}})\mathrm{d}y \geq 1-\alpha,\label{eq:naive_set_predictor_de_cond}
 \end{align} where the integral is approximated via Monte Carlo integration using samples $\hat{\mathbf{y}}^{1:M}$ drawn from the model $p(y|x,\theta_{\mathcal{D}})$.

\subsubsection{PQC Ansatz}
\label{subsubsec:pqc_ansatz_with_input}
As for the problem of density estimation, we adopt the \emph{hardware-efficient ansatz}, with the key difference that the parameterized unitary matrix $U(x|\theta)$ also encodes the input $x$. In this regard, we  follow the \emph{data reuploading} strategy introduced in \cite{perez2020data}, which  encodes the input $x$  across all layers of the PQC architecture. This has been shown to offer significant advantages in terms of model expressivity \cite{perez2020data}. We explore three different data encoding solutions in order of complexity: (\emph{i}) conventional \emph{fixed angle encoding} \cite{schuld2021effect}; (\emph{ii})  \emph{learned linear angle encoding} as studied in \cite{perez2020data}; and (\emph{iii})  \emph{learned non-linear angle encoding} which is related to the design recently studied in \cite{miao2023neural}.  

\begin{figure}
  \centering
  \includegraphics[width=13.5cm]{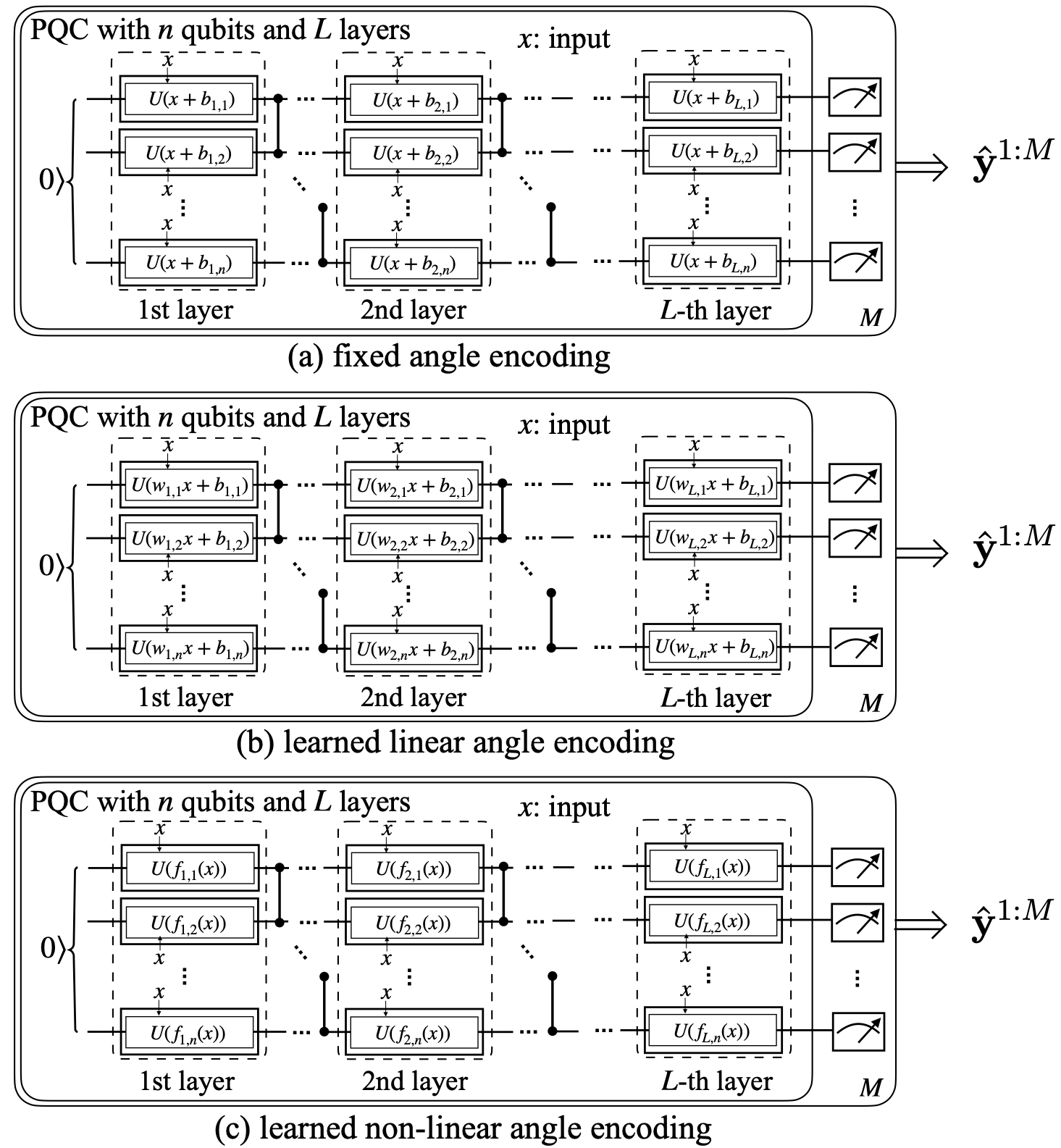}
  \caption{Illustration of the considered angle encoding strategies for the hardware-efficient ansatz adopted for the regression problem: (a) fixed angle encoding \cite{schuld2021effect} has no trainable parameters for the data encoding block; (b) learned linear angle encoding assumes trainable linear weight parameters that are multiplied with the input $x$ \cite{perez2020data}; (c) learned \emph{non-linear} angle encoding strategy maps the input $x$ to the gate-controlling angles via a classical neural network.} 
\label{fig:pqc_angle_encodings}
\end{figure}

The parameterized unitary matrix $U(x|\theta)$ with $L$ layers is defined as (cf. \eqref{eq:L_layer_unitary_PQC_without_input})
\begin{align}
    \label{eq:L_layer_unitary_PQC}
    U(x|\theta) = U_L(x|\theta)\cdot U_{L-1}(x|\theta)\cdots U_1(x|\theta),
\end{align}
with the unitary matrix $U_l(x|\theta)$ for each $l$-th layer modelled as (cf. \eqref{eq:parameterized_unitary_for_density_learning})
\begin{align}
    &U_l(x|\theta) = U_{ent}(R(\phi_{l,1}^1, \phi_{l,1}^2, \phi_{l,1}^3 ) \otimes \cdots \otimes R(\phi_{l,n}^1, \phi_{l,n}^2, \phi_{l,n}^3 ) ),
\end{align}
with $U_{ent}$ and $R(\phi^1, \phi^2, \phi^3)$ defined as in Sec.~\ref{subsubsec:pqc_ansatz_without_input}. Unlike Sec.~\ref{subsubsec:pqc_ansatz_without_input}, however, the three angles $[\phi_{l,k}^1, \phi_{l,k}^2, \phi_{l,k}^3]$ for $k$-th qubit at the $l$-th layer may encode information about the input $x$. Note that we did not indicate this dependence explicitly in the notation in order to avoid clutter. We set $L=5$. 

Angle encoding can be generally expressed with a parameterized (classical) function $f(\cdot|\theta_{l,k}):\mathcal{X}\rightarrow [0,2\pi]\times[0,2\pi]\times[0,2\pi]$ that takes as input $x$ and outputs three angles. The function is parameterized by a vector $\theta_{l,k}$, and is generally written as 
\begin{align}
    \label{eq:general_angle_encoding}
    f(x|\theta_{l,k}) = [\phi_{l,k}^1, \phi_{l,k}^2, \phi_{l,k}^3].
\end{align}
Accordingly, the parameter vector $\theta$ contains $n$ parameter vectors $\theta_{l,k}$ for each $l$-th layer, i.e., we have  
$\theta = \{\{ \theta_{l,k} \}_{k=1}^{n}\}_{l=1}^L$. As  mentioned, we consider 
three types of encoding functions, which are listed in order of increased generality.

Denoting as  $f^i(x|\theta_{l,k})$ the $i$-th output of the function $f(x|\theta_{l,k})$ in (\ref{eq:general_angle_encoding}),   \emph{conventional angle encoding} sets \cite{perez2020data}
\begin{align}
    \phi_{l,k}^i = f^i(x|\theta_{l,k}) =x + b_{l,k}^i,
\end{align}
where the model parameter vector encompasses the scalar biases $ \theta_{l,k} = \{ b_{l,k}^i  \}_{i=1}^3$.

\emph{Learning linear angle encoding} sets the function $f(x|\theta_{l,k})$ as \cite{perez2020data} 
\begin{align}
    \phi_{l,k}^i = f^i(x|\theta_{l,k}) = w_{l,k}^i x + b_{l,k}^i,
\end{align}
with parameters $\theta_{l,k} = \{w_{l,k}^i, b_{l,k}^i  \}_{i=1}^3$ containing scalar weights and scalar biases. 

Finally, \emph{learned non-linear angle encoding}   implements function $f(x|\theta_{l,k})$ as a multi-layer neural network. In this case, the parameter vector $\theta_{l,k}$ includes the synaptic weights and biases of the neural network. 
  In particular, we consider a fully connected network with input $x$ followed by two hidden layers with $10$ neurons having exponential linear unit (ELU) activation in each layer, that outputs the $3nL$ angles $\{ \{ \phi_{l,k}^1, \phi_{l,k}^2, \phi_{l,k}^3 \}_{k=1}^n \}_{l=1}^L$ for the PQC, i.e., $[ \phi_{1,1}^1,\phi_{1,1}^2,\phi_{1,1}^3,\phi_{1,2}^1,...,\phi_{1,n}^3,\phi_{2,1}^1,...,\phi_{L,n}^3]$.     We note that non-parametric non-linear angle encoding \cite[Sec.~6.8.1]{simeone2022introduction} and exponential angle encoding \cite{kordzanganeh2022exponentially} are also non-linear forms of angle encoding, which can be approximated via a suitable choice of the neural network function $f(x|\theta_{l,k})$ \cite{hornik1991approximation}.

\subsection{Training}
\label{subsec:pqc_training}
In this subsection, we elaborate on the implementation of training for the PQCs described in the previous subsections. We first note that the na\"ive predictors \eqref{eq:naive_set_predictor_de} and \eqref{eq:naive_set_predictor_de_cond} use the entire data set $\mathcal{D}$ for training, while QCP splits data set $\mathcal{D}$ into a training set $\mathcal{D}^{\text{tr}}$ and a calibration set $\mathcal{D}^{\text{cal}}$. We adopt an equal split, with $|\mathcal{D}^\text{tr}|=|\mathcal{D}|/2$ training examples and  $|\mathcal{D}^\text{cal}|=|\mathcal{D}|/2$ calibration examples, unless specified otherwise. For the rest of this subsection, we write the data used for training as $\mathcal{D}^\text{tr}$, with the understanding that for na\"ive schemes, this set is intended to be the overall set $\mathcal{D}$. We focus on the more general supervised learning setting, for which each training data example is a pair of input $x^\text{tr}$ and output $y^\text{tr}$, since for unsupervised learning it is sufficient to remove the input $x^\text{tr}$.

Given a training data set $\mathcal{D}^\text{tr}$, the PQC model parameter vector $\theta_{\mathcal{D}^\text{tr}}$ is trained using Adam \cite{kingma2014adam} based on gradients evaluated via automatic differentiation in PyTorch \cite{paszke2019pytorch} on a classical simulator. {\color{black} When training on a quantum hardware device, the parameter-shift rule \cite{schuld2019evaluating} is  utilized to update the parameters of the PQC.} In the rest of this subsection, we detail the loss functions used for the training of deterministic predictors \eqref{eq:deterministic_output_pqc} and implicit probabilistic models  \eqref{eq:single_measurement_distribution}.

\subsubsection{Training PQCs as Deterministic Predictors}
When the PQC is used as a deterministic predictor, the prediction is given by  $\hat{y}=\langle O \rangle_{\rho(x|\theta)}$ as in \eqref{eq:deterministic_output_pqc}. For this case,  we adopt the standard \emph{quadratic loss}, yielding the empirical risk minimization problem
\begin{align}
    \label{eq:quadratic_loss}
    \theta_{\mathcal{D}^\text{tr}} = \arg\min_\theta \sum_{i=1}^{|\mathcal{D}^\text{tr}|} (y^\text{tr}[i]-\langle O  \rangle_{\rho(x^\text{tr}[i]|\theta)})^2.
\end{align}

\subsubsection{Training PQCs as Implicit Probabilistic Models}
When the PQC is used as an implicit probabilistic predictor, with distribution $p(\hat{\mathbf{y}}=o_j|x,\theta)=\text{Tr}(\Pi_j\rho(x|\theta))$ as in \eqref{eq:single_measurement_distribution}, we adopt the cross-entropy loss. To this end, we first quantize the true label $y$ of an example $(x,y)$ so that the quantization levels match the possible values $\{o_j\}_{j=1}^{N'}$ obtained from the measurement of the given observable $O=\sum_{j=1}^{N'} o_j \Pi_j$. Then, the empirical risk minimization problem tackled during training is written as 
\begin{align}
    \label{eq:CE_loss}
    \theta_{\mathcal{D}^\text{tr}} = \arg\min_\theta \sum_{i=1}^{|\mathcal{D}^\text{tr}|}(-\log \text{Tr}( \Pi_{j_i} \rho(x^\text{tr}[i]|\theta) )),
\end{align}
where $j_i\in \{1,....,N'\}$ is the index of the eigenvalue $o_{j_i}$ that is closest (in Euclidean distance) to the $i$-th  training output  $y^\text{tr}[i]$.

%% file: Sections/7_Experimental_Results.tex
\section{Experimental Results}
\label{sec:experimental_results}
This section describes experimental results for both the unsupervised learning and supervised learning settings introduced in the previous section.

\subsection{Unsupervised Learning: Density Learning}
\label{subsec:unsup_density_learning}

As explained in the previous section, we compare the performance of {\color{black}conventional CP} (Sec.~\ref{sec:CCP}),   which requires an arbitrarily large number of shots, with the na\"ive predictor \eqref{eq:naive_set_predictor_de}, {\color{black} with the PCP predictor \cite{wang2022probabilistic},} and with the proposed QCP scheme (Sec. \ref{sec:qcp}), in terms of their coverage probability and of the average size of the predicted set.  These metrics are evaluated using $K=1000$ experiments as the averages $  \hat{\mathbf{P}}^\text{supp} = \frac{1}{K}\sum_{k=1}^K \int_{y \in \Gamma(\mathbfcal{D}_k^\text{cal},\theta_{\mathcal{D}^\text{tr}})} p^*(y)\mathrm{d}y$ and $\frac{1}{K}\sum_{k=1}^K |\Gamma(\mathbfcal{D}_k^\text{cal},\theta_{\mathcal{D}^\text{tr}})|$ for {\color{black}CP}; and as the averages $\hat{\mathbf{P}}_M^\text{supp} = \frac{1}{K}\sum_{k=1}^K \int_{y \in \mathbf{\Gamma}_M(\mathbfcal{D}_k^\text{cal},\theta_{\mathcal{D}^\text{tr}})} p^*(y)\mathrm{d}y$ and $\frac{1}{K}\sum_{k=1}^K | \mathbf{\Gamma}_M(\mathbfcal{D}_k^\text{cal},\theta_{\mathcal{D}^\text{tr}})|$ for {\color{black}PCP and} QCP. {\color{black} We use subscript $k$ to specify that the corresponding data is drawn for the $k$-th experiment (see Appendix~\ref{sec:finite_number_guarantee})}.

\begin{figure}
  \centering
  \includegraphics[width=0.9\textwidth]{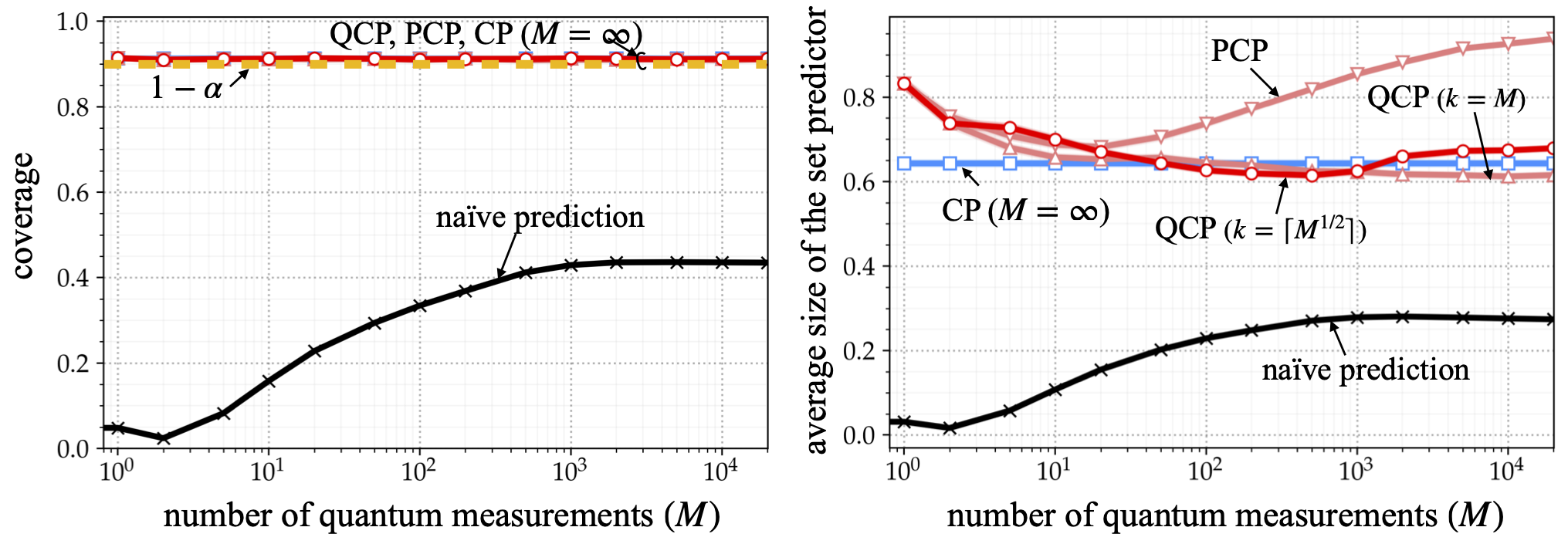}
  \caption{Density estimation with a weakly bimodal ground-truth Gaussian distribution: Coverage and average size of the set predictors as a function of number $M$ of quantum measurements given $|\mathcal{D}|=20$ available training samples. Training and testing are done on a classical simulator.   The shaded areas correspond to confidence intervals covering 95\% of the realized values.} 
\label{fig:density_weak}
\end{figure}

\subsubsection{A Visual Comparison}

Fig. \ref{fig:density_learning_intro} presents a visual comparison of the predicted sets produced by the different techniques with  $|\mathcal{D}|=20$ examples, assuming $M=200$ measurements from the PQC. Specifically, panel (e) in the figure displays examples of predicted sets obtained with the training data and with the  realizations of the calibration data set shown in panel (a) for $1-\alpha=0.9$. {\color{black} Conventional CP} is observed to fail when faced with a bimodal ground-truth distribution (shown in panel (a)). This is because it makes the underlying design assumption of a single-modal likelihood, as implied by the conventional choice of the quadratic loss \eqref{eq:quadratic_loss}. As a result, {\color{black}CP} produces inefficient set predictors. In contrast, a na\"ive set predictor tends to underestimate the coverage, since the samples produced by the trained PQC are excessively concentrated around the modes of the two Gaussians as seen in panels (b)-(d). {\color{black} While both PCP \cite{wang2022probabilistic} and QCP guarantee the validity condition \eqref{eq:validity_qcp}, it is observed that PCP} does not provide good performance in terms of efficiency of the set predictor in the presence of quantum hardware noise (see second row and third column of panel (e)), while  QCP ensures efficient set predictors.  This result suggests that increasing the value of $k$ enhances the robustness of the $k$-NN density estimator in \eqref{eq:pcp_scoring_from_generalized_with_k}, with $k=1$ yielding an excessive sensitivity to randomness due to shot and quantum noise.


\subsubsection{Quantitative Results with a Noiseless Simulator}

We now provide numerical evaluations of the coverage probability and of the average size of the set predictor as a function of the number of shots $M$ produced by the PQC. As in the example above, we also show the performance of {\color{black} conventional CP}, which assumes $M=\infty$. We assume a data set of  $|\mathcal{D}|=20$ data points, a target miscoverage level $\alpha=0.1$.

We start by assuming a \emph{weakly} bi-modal ground-truth distribution $p^*(y)$,  obtained as the mixture of Gaussians with similar means given their standard deviation. We specifically set $p^*(y) = \frac{1}{2}\mathcal{N}(-0.15, 0.1^2)+ \frac{1}{2}\mathcal{N}(0.15, 0.1^2)$. Fig.~\ref{fig:density_weak} shows 
coverage and average size of the set predictor as a function of number of quantum measurements $M$. While CP-based approaches provably provide well-calibrated support estimators, the na\"ive predictor \eqref{eq:naive_set_predictor_de} fails to cover the $1-\alpha$ of the mass of the ground-truth distribution $p^*(y)$. 

{\color{black} In terms of informativeness,} since the weakly bimodal distribution at hand can be well approximated by a unimodal Gaussian distribution, {\color{black} conventional CP} performs well. While using a finite number of samples, QCP performs comparably to CP as long as $M$ is not too small {\color{black}($M \geq 20$)}. {\color{black} In contrast, the set predictor produced by PCP tends to become uninformative with increased number $M$ of measurements ($M \geq 20$).} This is because, as illustrated in Fig.~\ref{fig:density_learning_intro}(b),  the trained PQC produces samples also in low-density regions of the ground-truth distribution $p^*(y)$, and the probability of obtaining one or more of such samples increases with the number of samples, $M$. {\color{black} This behavior becomes more evident by comparing QCP equipped with general scoring function  \eqref{eq:pcp_scoring_from_generalized_with_k} that uses $k=M$ -- labelled as ``QCP ($k=M$)'' in the figure -- to PCP.}

We now consider a ground-truth Gaussian distribution that presents a \emph{stronger} bimodality. To this end, we set $p^*(y) = \frac{1}{2}\mathcal{N}(-0.75, 0.1^2)+ \frac{1}{2}\mathcal{N}(0.75, 0.1^2)$, and the results are illustrated  in  Fig.~\ref{fig:density_strong}. While the conclusions in terms of coverage remain the same as in the previous example, comparisons in terms of average predicted set size reveal remarkably  different behaviors of the considered predictors. In particular, despite its requirements in terms of number of shots, {\color{black} conventional CP} significantly underperforms {\color{black} both PCP and QCP due to the underlying assumption of unimodality of the likelihood function, yielding a three times larger predicted set. Note that QCP equipped with general scoring function  \eqref{eq:pcp_scoring_from_generalized_with_k} that uses $k=M$ fails to perform well unlike the previous weakly bi-modal case, demonstrating the importance of our choice $k=\lceil M^{1/2} \rceil$ for QCP in \eqref{eq:closed_form_genreal_scoring_pcp_and_qcp}. Additional experimental results that study the impact of $k$ in the scoring function \eqref{eq:pcp_scoring_from_generalized_with_k} can be found in Appendix~\ref{appendix:additional_experiments}.}


\begin{figure}
  \centering
  \includegraphics[width=0.9\textwidth]{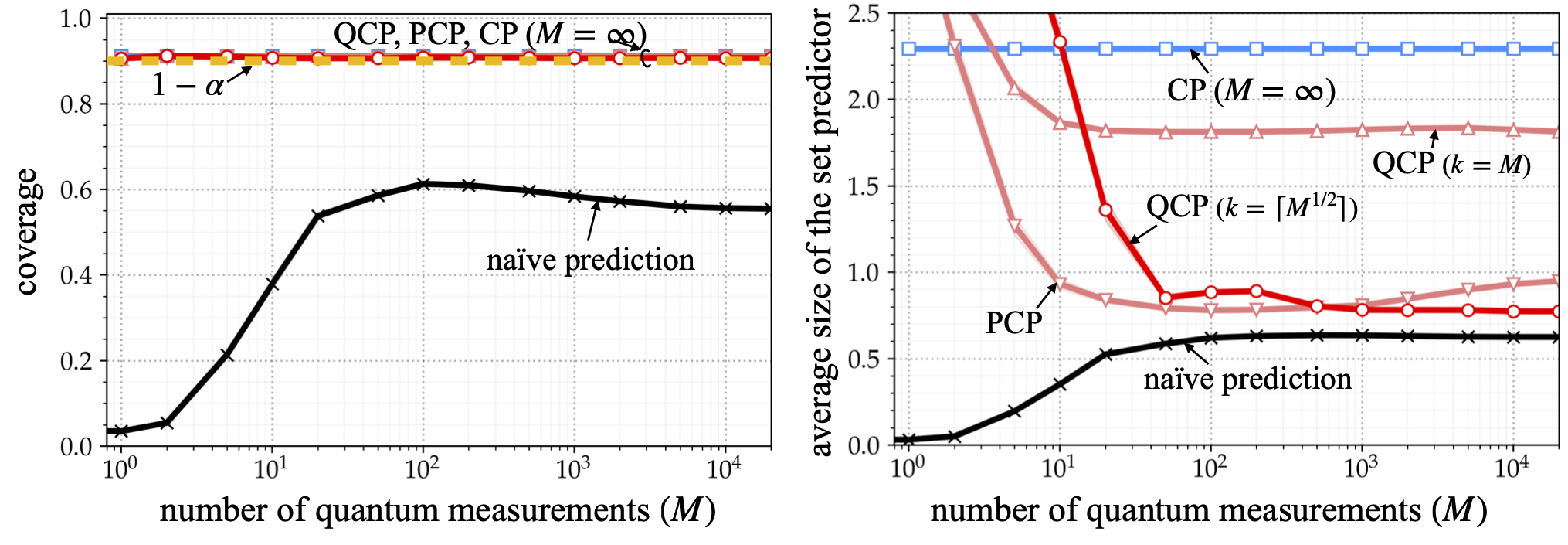}
  \caption{Density estimation with a strongly bimodal ground-truth Gaussian distribution: Coverage and average size of the set predictors as a function of number $M$ of quantum measurements given $|\mathcal{D}|=20$ available training samples. Training and testing are done on a classical simulator.   The shaded areas correspond to confidence intervals covering 95\% of the realized values.
  } 
\label{fig:density_strong}
\end{figure}

\begin{figure}
  \centering
  \includegraphics[width=0.9\textwidth]{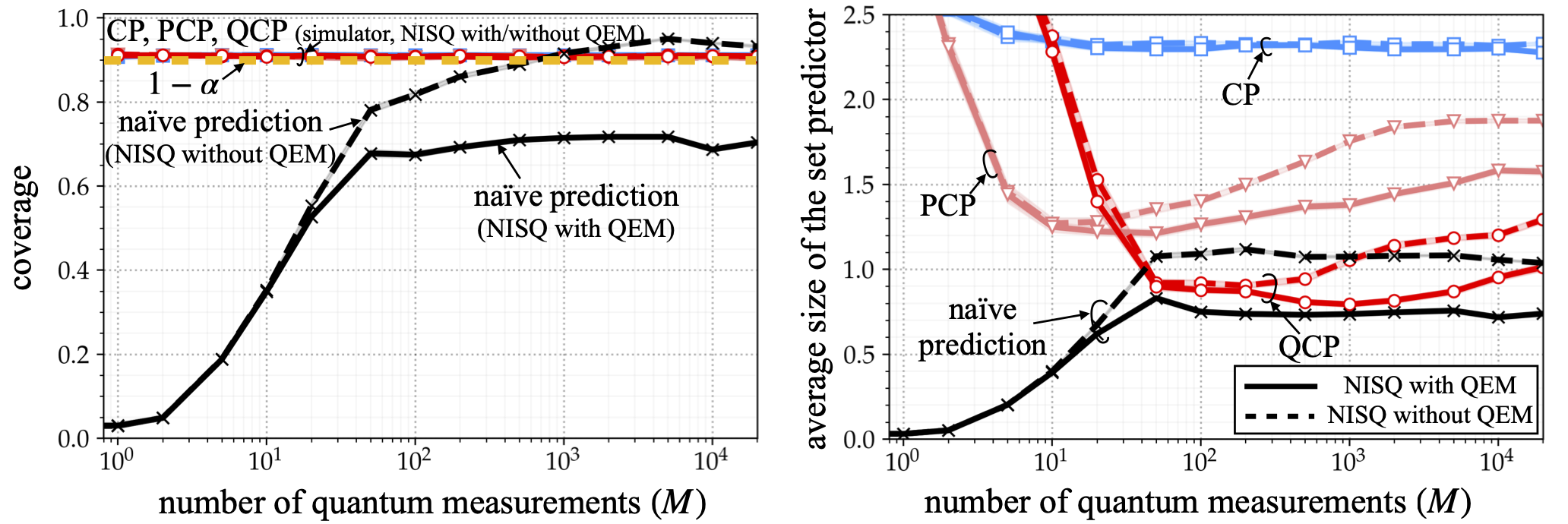}
  \caption{Density estimation with a strongly bimodal ground-truth Gaussian distribution: Coverage and average size of the set predictors as a function of number $M$ of quantum measurements given $|\mathcal{D}|=20$ available training samples.  Training is done on a classical simulator,  while testing is implemented on  \texttt{imbq\_quito} NISQ device, with or without M3 QEM \cite{nation2021scalable}. The shaded areas correspond to confidence intervals covering 95\% of the realized values.} 
\label{fig:density_ibmq}
\end{figure}

\subsubsection{Quantitative Results with Quantum Hardware Noise}
\label{subsubsec:experiment_ibmq_density}

Finally, we provide a quantitative performance comparison based on results obtained on the \texttt{imbq\_quito} NISQ device with or without  M3 quantum error mitigation (QEM) \cite{nation2021scalable}. As discussed in the previous section, for this case, training was done using a classical simulator, and  the quantum computer was used solely for testing, i.e., to produce the density support estimate. {\color{black} Additional experimental results with PQC trained on \texttt{imbq\_quito} NISQ device can be found in Appendix~\ref{appendix:additional_experiments}. } For reference, for this experiment, we show the performance of {\color{black} conventional CP} by using the empirical average of the measurements, i.e., $\frac{1}{M}\sum_{m=1}^M \hat{\mathbf{y}}^{m}$,  in lieu of the true expectation $\hat{y}=\langle O \rangle_{\rho(x|\theta_{\mathcal{D}^\text{tr}})}$. This allows to report results for CP that depend on the number of shots, $M$.

Fig.~\ref{fig:density_ibmq} validates the conclusion reported in Sec.~\ref{sec:QCP} that {\color{black} both PCP and  QCP are} provably well calibrated despite the presence of  quantum hardware noise. In contrast, the na\"ive predictor is not well calibrated, even in the absence of quantum hardware noise  (see  Fig.~\ref{fig:density_strong}). In this regard, the na\"ive predictor is observed to benefit from quantum hardware noise, achieving validity with a sufficiently large number of measurements $M\geq 10^3$. This can  be understoood as a  consequence of the fact that quantum hardware noise tends to produce samples that cover a wider range of output values \cite{nation2021scalable}, making it  easier to cover a fraction $1-\alpha$ of the mass of the original density $p^*(y)$ using the predictor \eqref{eq:naive_set_predictor_de}.

In terms of average size of the prediction set, {\color{black}PCP} is seen to be particularly sensitive to quantum hardware noise, as also anticipated with the examples in Fig.~\ref{fig:density_learning_intro}. In this case, increasing the number of shots can be deleterious as more noise is injected in the estimate of the predicted set. In contrast, {\color{black} QCP} set predictor provides a significantly more robust performance to quantum hardware noise,  even in the presence of QEM. {\color{black}Finally,} QEM is observed to improve the informativeness of both the {\color{black} PCP and QCP} set predictors by enhancing the quality of the underlying probabilistic predictor via the mitigation of quantum noise (see, e.g., \cite{dhillon2023expected}).

\subsection{Supervised Learning: Regression}
\label{subsec:supervised_learning_regression}
In this subsection, we turn to the supervised learning problem with ground-truth distribution \eqref{eq:regression_target}. As done for unsupervised learning, we compare QCP (Sec.~\ref{sec:QCP})  against {\color{black} conventional CP} -- in the ideal case of an infinite shots (Sec.~\ref{sec:CCP}) -- as well as against the na\"ive predictor \eqref{eq:naive_set_predictor_de_cond}  {\color{black} and the PCP predictor \cite{wang2022probabilistic}} by evaluating coverage probability and average size of the set predictors using $K=500$ experiments as discussed {\color{black} after Theorem~\ref{ther:qcp} and detailed in  Appendix \ref{sec:finite_number_guarantee}.}

\begin{figure}
  \centering
  \includegraphics[width=0.9\textwidth]{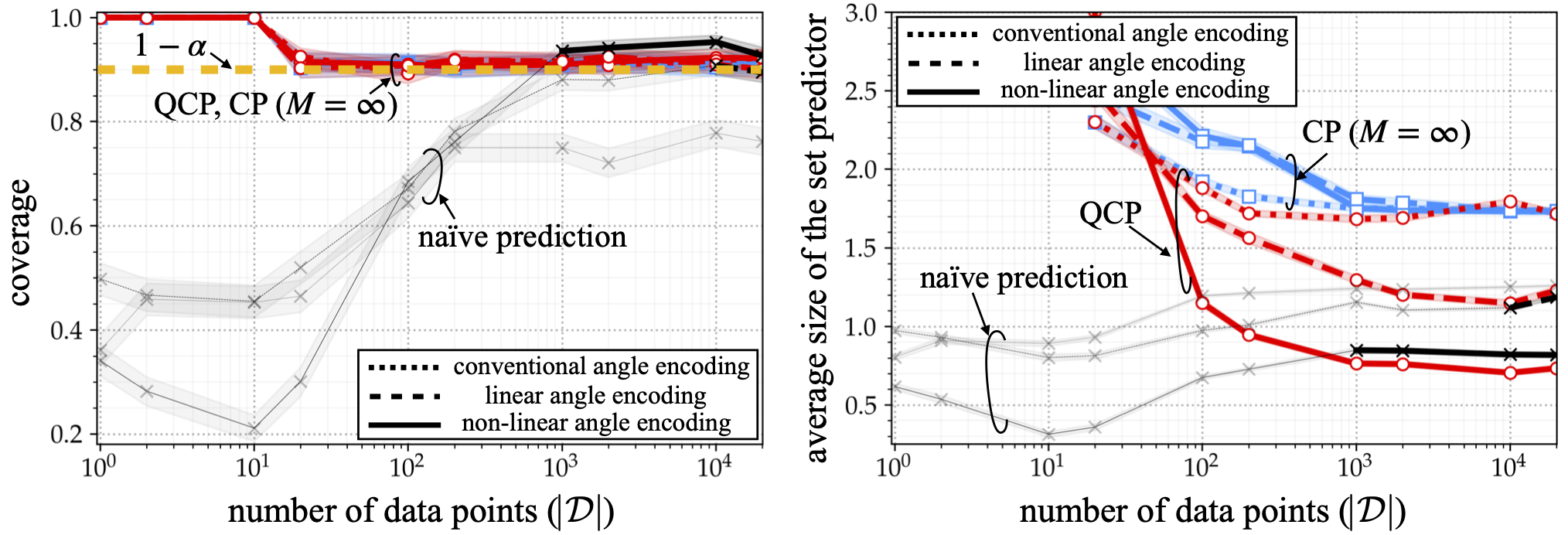}
  \caption{Regression for mixture of two sinusoidal functions \eqref{eq:regression_target}: Coverage and average size of the set predictors as a function of the number $|\mathcal{D}|$ of available data points given $M=1000$ quantum measurements. Training and testing are done on a classical simulator. The shaded areas correspond to confidence intervals covering 95\% of the realized values. The results are averaged over $K=1000$ experiments, and the transparent lines are used for set predictors that do not meet the coverage level $1-\alpha=0.9$. } 
\label{fig:regression_per_example_many_shots}
\end{figure}
\begin{figure}
  \centering
  \includegraphics[width=0.9\textwidth]{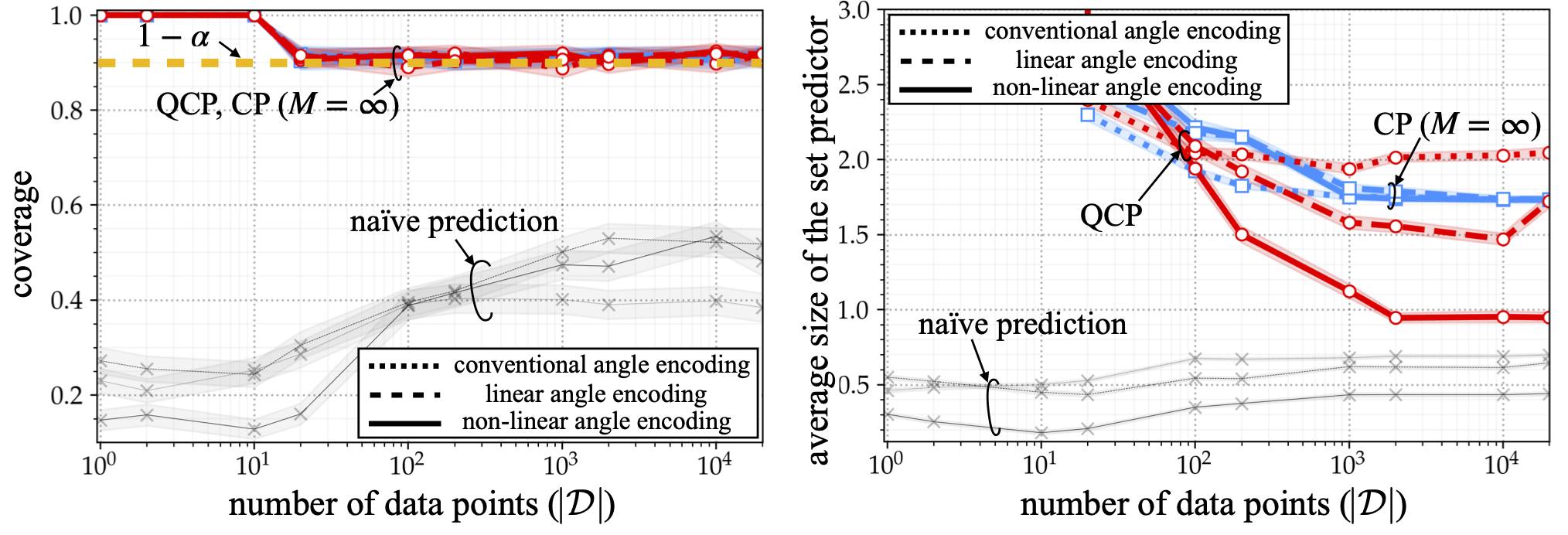}
  \caption{Regression for mixture of two sinusoidal functions \eqref{eq:regression_target}: Coverage and average size of the set predictors as a function of the number $|\mathcal{D}|$ of available training samples given $M=20$ quantum measurements. Training and testing are done on a classical simulator. The shaded areas correspond to confidence intervals covering 95\% of the realized values.  The results are averaged over $K=1000$ experiments, and the transparent lines are used for set predictors that do not meet the coverage level $1-\alpha=0.9$.} 
\label{fig:regression_per_example_few_shots}
\end{figure}

\subsubsection{A Visual Comparison}
Fig.~\ref{fig:regression_learning_intro} presents a visual comparison of the different set predictors assuming a data set of $|\mathcal{D}|=100$ data points with $M=100$ shots.  We adopt learned non-linear angle encoding as described in Sec.~\ref{subsubsec:pqc_ansatz_with_input} with neural network composed of two hidden layers, each with $10$ neurons having ELU activations. Panel (b) in the figure depicts examples of predicted sets obtained with the training data and with the  realizations of the calibration data set shown in panel (a) for $1-\alpha=0.9$. The na\"ive set predictor is observed to  underestimate the support of the distribution. This can be interpreted in light of the typical overconfidence of trained predictors \cite{kobayashi2022overfitting}, which causes  the na\"ive predictor \eqref{eq:naive_set_predictor_de_cond} to concentrate on a smaller subset of values as compared to the desired support, at level $1-\alpha=0.9$, of the ground-truth distribution. 
In contrast, {\color{black}both PCP and QCP provably satisfy} the predetermined coverage level of $1-\alpha=0.9$, {\color{black} while PCP yields uninformative set predictor, especially in the presence of quantum hardware noise, even with QEM \cite{nation2021scalable}. In contrast, QCP can successfully suppress quantum hardware noise to yield more informative set as compared to PCP,  as seen in the last column of panel (b).} 

\begin{figure}
  \centering
  \includegraphics[width=0.9\textwidth]{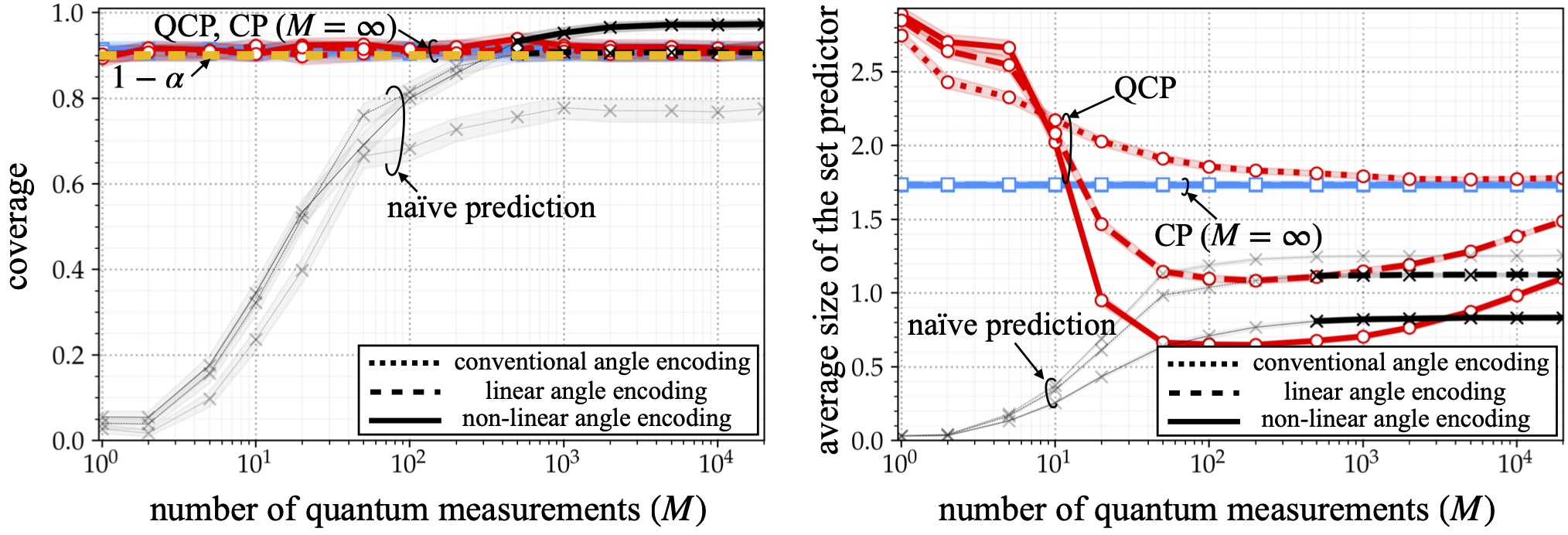}
  \caption{Regression for mixture of two sinusoidal functions \eqref{eq:regression_target}: Coverage and average size of the set predictors as a function of number $M$ of quantum measurements given $|\mathcal{D}|=10^4$ available training samples. Training and testing are done on a classical simulator. The shaded areas correspond to confidence intervals covering 95\% of the realized values.  The results are averaged over $K=1000$ experiments, and the transparent lines are used for set predictors that do not meet the coverage level $1-\alpha=0.9$.} 
\label{fig:regression_per_measurements_many_training}
\end{figure}

\begin{figure}
  \centering
  \includegraphics[width=0.9\textwidth]{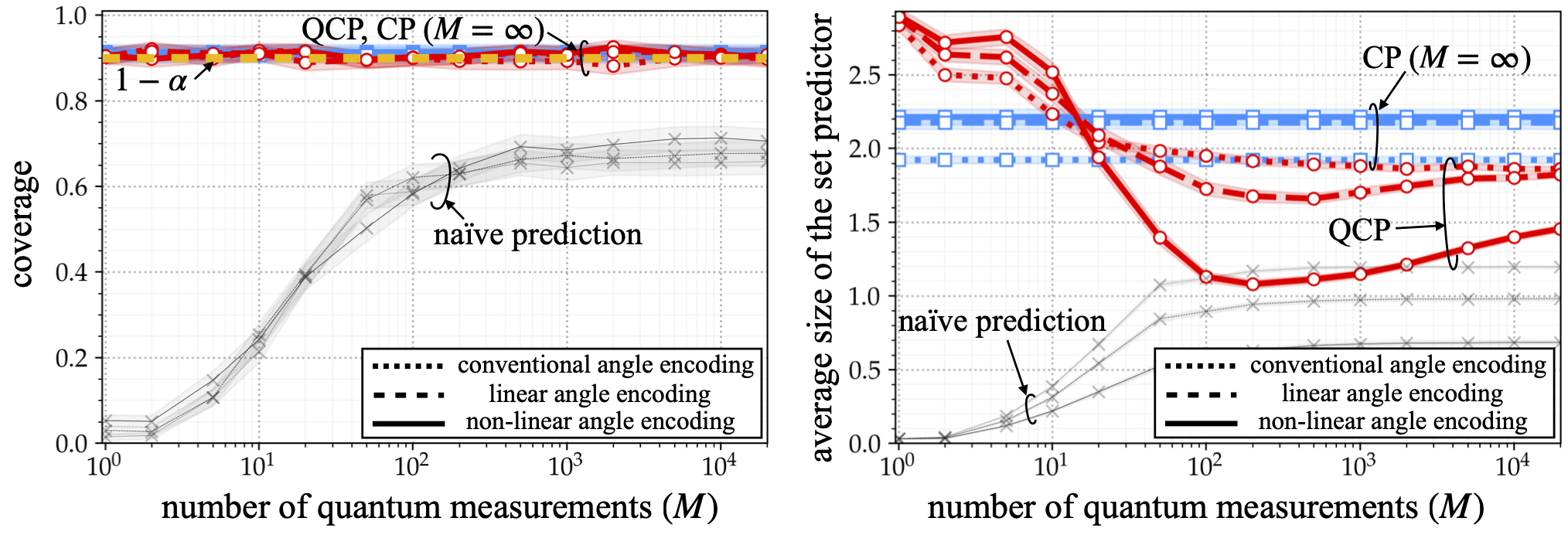}
  \caption{Regression for mixture of two sinusoidal functions \eqref{eq:regression_target}: Coverage and average size of the set predictors as a function of number $M$ of quantum measurements given $|\mathcal{D}|=100$ available training samples. Training and testing are done on a classical simulator. The shaded areas correspond to confidence intervals covering 95\% of the realized values.  The results are averaged over $K=1000$ experiments, and the transparent lines are used for set predictors that do not meet the coverage level $1-\alpha=0.9$. } 
\label{fig:regression_per_measurements_few_training}
\end{figure}

\subsubsection{Quantitative Results with a Noiseless Simulator}

We now provide numerical evaluations of the coverage probability and of the average size of the set predictor as a function of the number $|\mathcal{D}|$ of data points for different angle encoding strategies as described in Sec.~\ref{subsubsec:pqc_ansatz_with_input}. For learning non-linear angle encoding, we adopt the same architecture described above. For CP set predictors, if $|\mathcal{D}|\leq 20$, we split the data set $\mathcal{D}$ in equal parts for training and calibration; while, when $|\mathcal{D}|> 20$, we fix the number of calibration examples to $|\mathcal{D}^\text{cal}|=10$ to use all the remaining data points for training. We assume a target miscoverage level $\alpha=0.1$, and the reported values of  coverage and average size of the set predictor are averaged over $K=1000$ experiments as {\color{black} discussed after Theorem~\ref{ther:qcp} and detailed in  Appendix \ref{sec:finite_number_guarantee}.}

In Fig.~\ref{fig:regression_per_example_many_shots} and  Fig.~\ref{fig:regression_per_example_few_shots}, we show the coverage rate and average set size as a function of the number of data points, $|\mathcal{D}|$, with a large and small numbers of shots, namely $M=1000$ and $M=20$, respectively.  In the first case, with a larger $M$, given enough training examples, here, for $|\mathcal{D}|\geq 10^4$, na\"ive prediction yields a well-calibrated set predictor that achieves $1-\alpha$ coverage. This is the case when adopting either linear \cite{perez2020data} or non-linear angle encoding. In contrast, na\"ive prediction fails to meet the coverage requirements for smaller data set size; and also for the smaller value of $M$ irrespective of the data set size. 

In line with their theoretical properties, {\color{black} conventional CP} and QCP are guaranteed to provide  coverage at the desired level  $1-\alpha$, irrespective of data availability and number of shots. However, despite the use of an arbitrarily large number of shots, {\color{black} CP} produces larger prediction set sizes than QCP. As in the case of unsupervised learning studied in the previous subsection, this is caused by the unimodality of the likelihood function assumed by {\color{black} CP}. As an example, given $|\mathcal{D}|=2000$ with non-linear angle encoding, when $M=1000$, QCP yields set predictors with average size $0.76$, while the average size of {\color{black} CP} predictors is $1.74$; and with $M=20$ we obtain set size $1.74$ with {\color{black} CP} and $0.94$ for QCP.

We now further investigate the impact of the number $M$ of shots by focusing on regimes with abundant data, i.e., with $|\mathcal{D}|=10^4$, and with limited data, i.e., $|\mathcal{D}|=100$ in Fig.~\ref{fig:regression_per_measurements_many_training} and Fig.~\ref{fig:regression_per_measurements_few_training}, respectively. In a manner that parallels the discussion in the previous paragraph on the role of the data set size, if $|\mathcal{D}|$ is sufficiently large, na\"ive set prediction achieves the desired coverage level as long as the number of shots is also sufficiently large, here $M\geq 500$ (Fig.~\ref{fig:regression_per_measurements_many_training}). In contrast, CP schemes attain calibration in all regimes, with QCP significantly outperforming {\color{black}CP} when $M$ is not too small. As an example, given $M=200$ with non-linear angle encoding, when $|\mathcal{D}|=10^4$, QCP yields set predictors with average size $0.65$, while the average size of  {\color{black} conventional CP} predictors is $1.73$; and with $|\mathcal{D}|=100$ we obtain set size $2.21$ with {\color{black}CP} and $1.08$ for QCP.  Following the discussion in the previous subsection, an increase in the number of shots $M$ does not necessarily translate in more efficient set predictors,  as it becomes more likely for the PQC to draw outliers that cover low-density areas of the ground-truth distribution.

\begin{figure}
  \centering
  \includegraphics[width=0.9\textwidth]{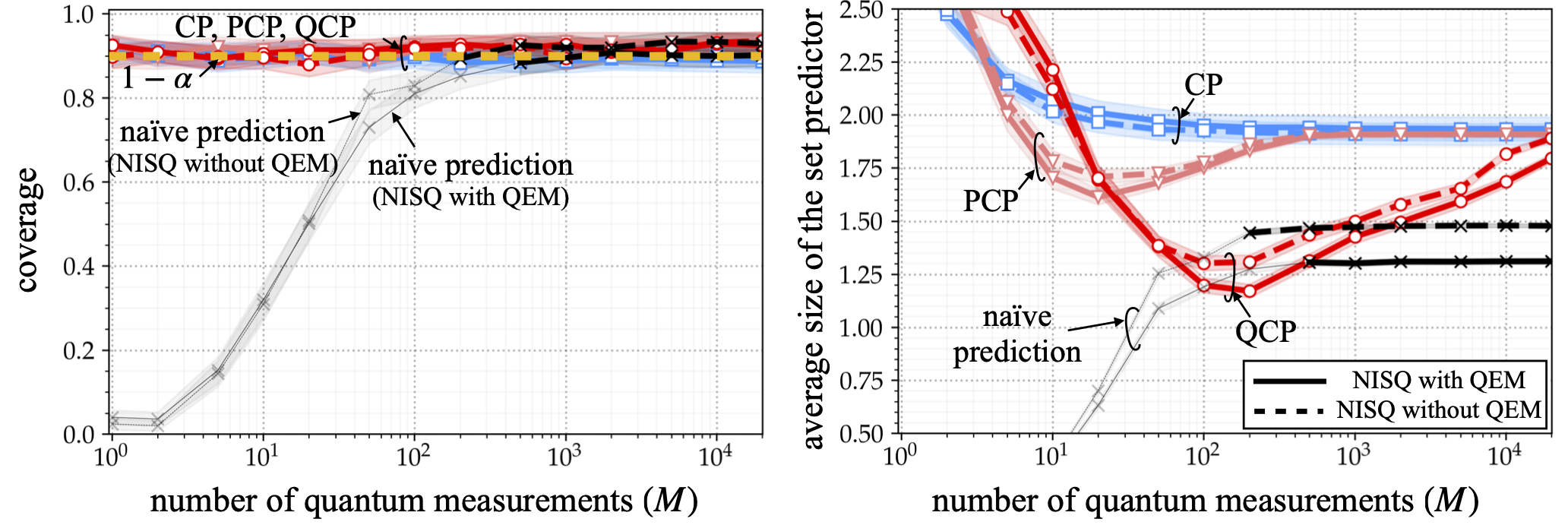}
  \caption{Regression for mixture of two sinusoidal functions \eqref{eq:regression_target}: Coverage and average size of the set predictors as a function of number $M$ of quantum measurements given $|\mathcal{D}|=100$ available training samples. Training is done on a classical simulator,  while testing is implemented on  \texttt{imbq\_quito} NISQ device, with or without M3 QEM \cite{nation2021scalable}. The shaded areas correspond to confidence intervals covering 95\% of the realized values.  The results are averaged over $K=500$ experiments, and the transparent lines are used for set predictors that do not meet the coverage level $1-\alpha=0.9$. } 
\label{fig:regression_per_measurements_few_training_ibmq}
\end{figure}

\subsubsection{Quantitative Results with Quantum Hardware Noise}
Moreover, we present  a quantitative performance comparison based on results obtained on the \texttt{imbq\_quito} NISQ device with or without  M3 QEM. As discussed in the previous section, for this case, training was done using a classical simulator, and  the quantum computer was used solely for testing, i.e., to produce the set prediction given a test input. In Fig.~\ref{fig:regression_per_measurements_few_training_ibmq}, we investigate the performance metrics as a function of number of quantum measurements $M$. QCP is confirmed to provide the best performance within a suitable range of values of $M$, significantly outperforming {\color{black} both PCP \cite{wang2022probabilistic} and conventional CP,}  even in the presence of quantum noise, with or without QEM.

{\color{black}


\subsection{Quantum Data Classification}\label{sec:exp_qc}
Finally,  we consider the problem of classifying quantum states  described in Sec.~\ref{sec:application_to_quantum}.
In particular, we aim at classifying $C=10$ density matrices of size $16\times 16$, assuming a uniform label probability $p(y)=1/C$ for all $y\in\{1,...,10\}$. The density matrix $\rho(y)$ for each class $y$ is expressed as the Gibbs state $\rho(y) = e^{-H(y)/T}/\text{Tr}(e^{-H(y)/T})$ with temperature $T>0$, where the Hamiltonian matrices $H(y)$ are independently generated so as to ensure a sparsity level of $0.2$ at temperature $T=1$ as in reference  \cite{ahmed2021classification}. We adopt a \emph{pretty good measurement} detector \cite{hausladen1994pretty, gambs2008quantum} as the pre-designed POVM $\mathcal{P}$. {\color{black}We note that we could have also adopted PQC-based detectors designed using a number of copies of the training states in a manner similar to reference \cite{deville2021new}, since the validity condition \eqref{eq:qcp_QC_main} holds for any fixed detector.}


\begin{figure}[t]
  \centering
  \includegraphics[width=0.95\textwidth]{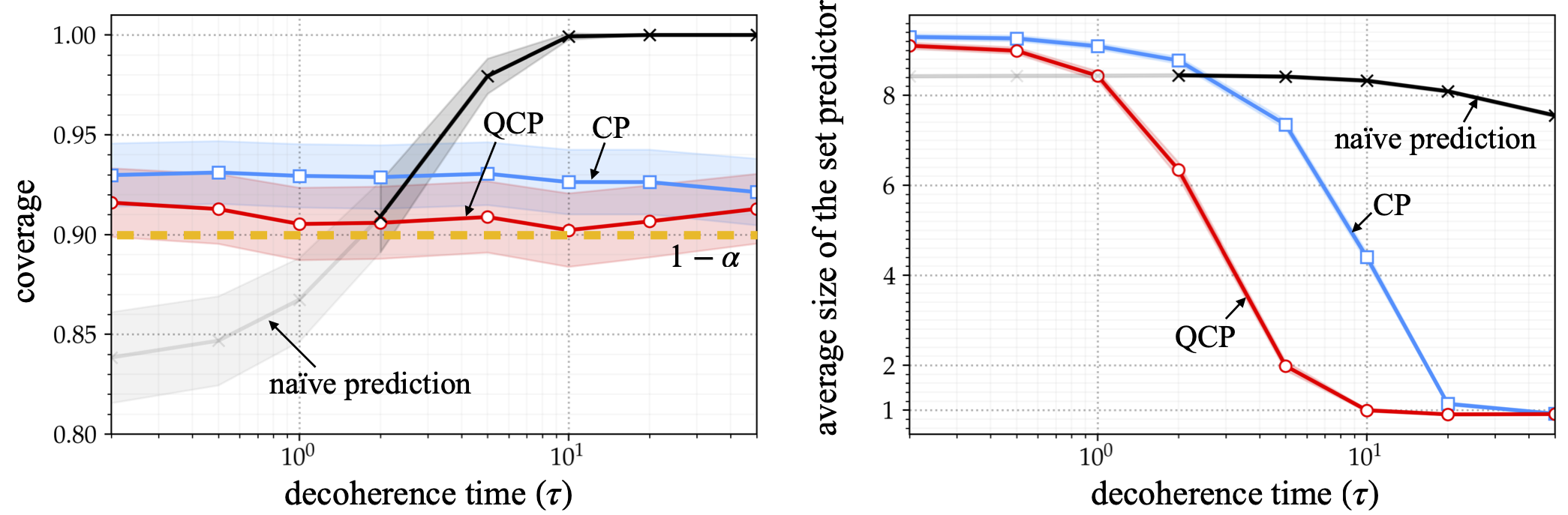}
  \caption{Quantum data classification under drift \cite{schwarz2011detecting}: Coverage and average size of the set predictors as a function of the decoherence time $\tau$ in \eqref{eq:decoherence_evolution} given $|\mathcal{D}^\text{cal}|=10$ calibration samples.  {\color{black} The ten possible density matrices to be classified are generated as $\rho(y) = e^{-H(y)/T}/\text{Tr}(e^{-H(y)/T})$ with temperature $T>0$, where the Hamiltonian matrices $H(y)$ are independently generated so as to ensure a sparsity level of $0.2$ at temperature $T=1$ as in \cite{ahmed2021classification}. Pretty good measurements detector is adopted \cite{hausladen1994pretty}. The shaded areas correspond to confidence intervals covering 95\% of the realized values.  The results are averaged over $K=1000$ experiments, and transparent lines are used to highlight regimes in which the  set predictors do not meet the coverage level $1-\alpha=0.9$.}} 
\label{fig:qc_exp_drift}
\vspace{-0.2cm}
\end{figure}

In order to study the impact of drift, we apply the noise model \eqref{eq:decoherence_evolution} to the density matrix $\rho(y)$, i.e., we replace $\rho(x|\theta_{\mathcal{D}^\text{tr}})$  with $\rho(y)$ in \eqref{eq:decoherence_evolution}. We further assume the noise density matrix $\rho^\mathcal{N}(y)$ in \eqref{eq:decoherence_evolution} to be the fully mixed state $I/16$ as in reference \cite{schwarz2011detecting}, and set the temperature as $T=0.01$. In Fig.~\ref{fig:qc_exp_drift}, we plot coverage and average size of the set predictor as a function of the decoherence time parameter $\tau$ in \eqref{eq:decoherence_evolution} for  CP, QCP, and for the na\"ive predictor for $M=100$ shots. CP and the na\"ive predictor are implemented by using the estimated histogram, while PCP is not applicable for classification problems \cite{wang2022probabilistic}.  Unlike CP and the na\"ive set predictor, QCP builds on the weighted histogram (Sec.~\ref{sec:qcp}), and we choose the weights $w_m$ in \eqref{eq:general_non_para_DE} as $w_m \propto e^{-m/\tau}$.

As per our theoretical results, CP and QCP always achieve coverage no smaller than the predetermined level $1-\alpha=0.9$, irrespective of the amount of drift, while the na\"ive prediction fails to achieve validity for strong drifts ($\tau < 2$).  Furthermore,  QCP significantly reduces the size of the set predictor as compared to CP by taking into account the decreasing  quality of measurement shots via the use of the weighted histogram. For instance, given decoherence time $\tau=10$, QCP yields set predictor with average size  $1$, which is $4.4$ times smaller than the CP set predictor. 

%% file: Sections/8_Conclusion.tex
\section{Conclusions}
\label{sec:conclusion}
In this paper, we have proposed a general methodology for QML, referred to as \emph{quantum conformal prediction} (QCP), that  provides ``error bars'' with coverage guarantees that hold irrespective of the size of the training data set, of the number of shots, and of the type, correlation, and drift of  quantum hardware noise. Experimental results have shown that QCP can reduce the size of the error bars by up to eight times as compared to existing baselines,  and up to four times as compared to a direct application of CP to quantum models. Furthermore, it is concluded that, when quantum models are augmented with QCP, it is generally advantageous not to average over the shots, as typically done in the literature. Rather, treating the shots as separate samples allows QCP to obtain more informative error bars. Future directions for research include the generalization of the QCP framework to more general form of risk control beyond coverage \cite{bates2021distribution, angelopoulos2021learn, angelopoulos2022conformal}.

%% file: Sections/9_appendix.tex
\appendices
\section{Generalization Analysis vs. Conformal Prediction}
\label{supp:sec:gen_vs_cp}
\label{appendix:gen_bounds}
In this work, we have focused on quantifying prediction uncertainty caused by noisy quantum models trained based on limited data when applied on any new input. Generalization analysis \cite{shalev2014understanding, simeone2022machine} also studies the model performance on test inputs. However, as we briefly discuss  below,  it typically fails to provide meaningful operational error bars, unlike QCP. 

As illustrated in Fig. \ref{fig:overall}, \emph{generalization analysis} focuses on the identification of \emph{analytical} scaling laws on the amount of data required to ensure desired performance levels on test data \cite{shalev2014understanding}, possibly as a function of the training algorithm \cite{mcallester1999pac, guedj2019primer} and of the data distribution \cite{russo2016controlling, xu2017information} (see also \cite{simeone2022machine}). Related studies have also been initiated for quantum machine learning, with recent results including \cite{caro2022generalization,banchi2021generalization, jose2022transfer, abbas2021power, weber2022toward}, as summarized by \cite{banchi2023statistical}.

As a notable example, reference \cite{caro2022generalization} reveals the important insight that the \emph{generalization error}, i.e., the discrepancy between training and test losses, for quantum models grows as the square root of the number of gates $T$ and with the inverse of the square root of the size of the data set. While critical to gauge the feasibility to train quantum circuits, such results provide limited \emph{operational} guidelines concerning the uncertainty associated with decisions made on test data when training data are limited (see Fig.~\ref{fig:overall}).

Furthermore, as mentioned in the main text, most papers on quantum machine learning, including on generalization analyses, treat the output of a quantum model as \emph{deterministic}, implicitly assuming that expected values of observables can be calculated exactly \cite{caro2022generalization,banchi2021generalization, jose2022transfer, abbas2021power, weber2022toward}. In practice, for this to be an accurate modelling assumption, one needs to carry out a sufficiently large -- strictly speaking, infinite -- number of measurements, or \emph{shots}, at the output of the quantum circuit (see Fig.~\ref{fig:overall}, left column). These measurements are averaged to obtain the final prediction. From a statistical viewpoint, this assumption conveniently makes a quantum  models \emph{likelihood-based}, in the sense that one can evaluate exactly the probability of any output given an input and the model parameters (see Sec.~\ref{sec:CCP}). 

However,  as described in Sec.~\ref{sec:PQC_as_implicit} of the main text, quantum models are more properly described as being \emph{implicit}, or \emph{simulation-based}, in the sense that they only provide random samples from a given, inaccessible, likelihood \cite{duffield2022bayesian}. The generalization capabilities of \emph{generative} quantum models -- an example of implicit models -- have been recently studied in a separate line of research, including in  \cite{du2022theory, gili2022evaluating}.

In contrast to generalization analysis, CP \cite{vovk2022algorithmic} does not aim at obtaining analytical conclusions concerning sample efficiency. Rather, it provides a  general methodology to obtain error bars with formal guarantees on the probability of covering the correct test output (see Sec. ~\ref{sec:CCP} and Sec.~\ref{sec:QCP}). Such guarantees hold irrespective of the size of the training data set \cite{angelopoulos2021gentle, vovk2022algorithmic, barber2021predictive} (see Fig. \ref{fig:overall}, right column). 


\begin{figure}
  \centering
  \includegraphics[width=0.65\textwidth]{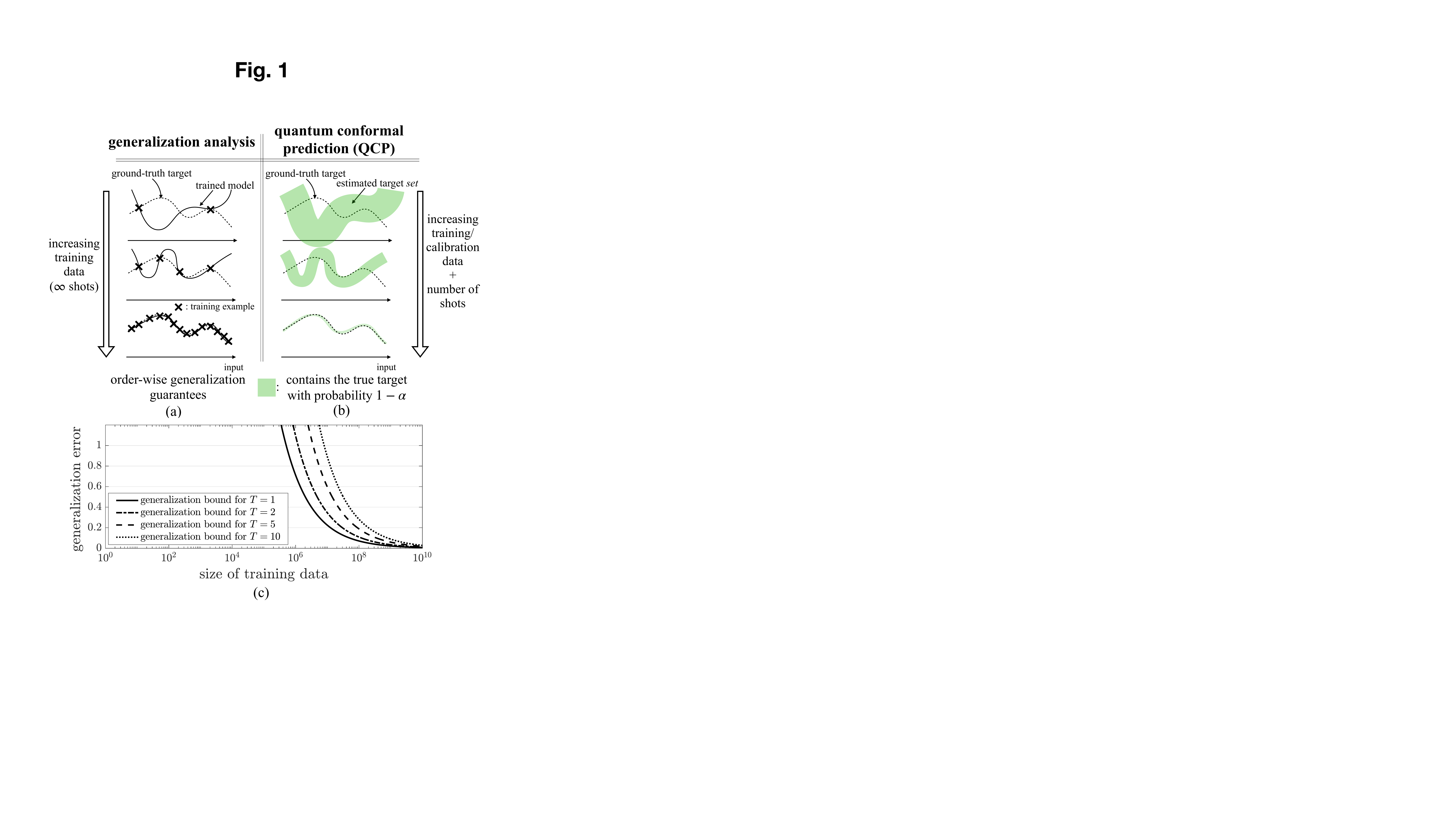}
  \caption{Comparison between (a)  quantum generalization analysis \cite{caro2022generalization,banchi2021generalization, jose2022transfer, abbas2021power, weber2022toward} and (b) quantum conformal prediction (QCP), which is introduced in this work. Quantum generalization analysis provides \emph{analytical} bounds on the generalization error that explicitly capture the general dependence on the number of training examples (typically assuming an infinite number of shots). Based on such bounds, one is able to conclude that, if the number of data points scales sufficiently quickly with respect to the model complexity, the trained model generalizes well outside the training data. In contrast, QCP provides an \emph{operational} way of quantifying the uncertainty of the decisions made outside the training set (shaded areas). The resulting ``error bars'' are  \emph{guaranteed} to contain the ground-truth output with a desired probability, irrespective of the amount of data, of the of the size of the training data set, of the ansatz of the QML model, of the training algorithm, of the number of shots, and of the type, correlation, and drift of quantum hardware noise. (c)  As an illustration of the results that can be obtained via generalization analysis, this panel shows the generalization bounds derived in \cite{caro2022generalization} as a function of number of training examples for different numbers $T$ of trainable local quantum gates (see Sec~\ref{appendix:gen_bounds} for details). As suggested by the plot, while very useful to identify general trends and scaling laws, generalization analyses only provide numerically meaningful bounds with a very large number of training examples. \label{fig:overall}}
\end{figure}

Details on the generalization bounds plotted in Fig.~\ref{fig:overall}(a) are provided next. As discussed in Sec. \ref{supp:sec:gen_vs_cp}, the figure reports the generalization bounds derived in \cite{caro2022generalization} as a function of the size of training data, $|\mathcal{D}^\text{tr}|$, for PQCs with $T=1,2,5,10$ trainable local quantum gates.

Let us fix a loss function $l(x,y|\theta)$ of the form  $l(x,y|\theta)=\text{Tr}(O_y^\text{loss}\rho(x|\theta)) $ for any example $(x,y)$, where $O_y$ is the \emph{loss observable} for target variable value $y$. We recall that $\rho(x|\theta)$ is the density matrix that describes the state produced by the PQC with model parameters $\theta$. As an example, choosing the loss observable as $O_y^\text{loss}= \sum_{j}\mathbbm{1}(o_j \neq y)\Pi_{j} $, with $\Pi_j = |j\rangle \langle j|$, ensures that the loss function $l(x,y|\theta)$ measures the \emph{probability of error}, i.e., the probability that the measurement output of the PQC is not equal to the label $y$.

We are interested in bounding the \emph{generalization error}, which is given by the difference between the \emph{population loss} $\mathbb{E}_{(\mathbf{x},\mathbf{y})\sim p(x,y)}[l(\mathbf{x},\mathbf{y}|\theta)]$, evaluated with respect to the ground-truth, unknown, distribution $p(x,y)$ and the training loss $1/|\mathcal{D}^\mathrm{tr}|\sum_{(x,y)\in \mathcal{D}^\text{tr}}l(x,y|\theta)$. Reference \cite{caro2022generalization} showed that the following  bound holds with probability at least $1-\delta$ over the choice of training data set $\mathcal{D}^\text{tr}$ \cite[Theorem 6]{caro2022generalization}
\begin{align}
    \label{eq:app_gen_caro}\mathbb{E}_{(\mathbf{x},\mathbf{y})\sim p(x,y)}[l(\mathbf{x},\mathbf{y}|\theta)]-& \frac{1}{|\mathcal{D}^\mathrm{tr}|}\sum_{(x,y)\in \mathcal{D}^\text{tr}}l(x,y|\theta) \\ \nonumber
    \leq \frac{24C_{\text{loss}}}{\sqrt{|\mathcal{D}^\text{tr}|}}\sqrt{512T}  \bigg( \frac{1}{2}&\sqrt{\log(6T)} + \frac{1}{2}\sqrt{\log 2} - \frac{\sqrt{\pi}}{2} \text{erf}(\sqrt{\log2}) + \frac{\sqrt{\pi}}{2} \bigg) + 3C_\text{loss}\sqrt{\frac{2\log(2/\delta)}{|\mathcal{D}^\text{tr}|}},
\end{align}where $T$ is the number of trainable gates in the PQC and $C_{\text{loss}}$ is the maximum spectral norm of the loss observables $O_y^\text{loss}$. The error function in \eqref{eq:app_gen_caro} is defined as $\text{erf}(x)=\frac{2}{\sqrt{\pi}} \int_{0}^x \exp{(-t^2)}\mathrm{d}t$. Fig.~\ref{fig:overall}(c) plots this bound for different values of $|\mathcal{D}^\text{tr}|$ and $T$ by assuming the probability of error loss described above, which has $C_{\text{loss}}=1$. We choose $\delta=0.1$, but changes in $\delta$ have a negligible impact on the bound. 

\section{Conformal Prediction for Classical Probabilistic Models (PCP)}
\label{sec:probabilistic_CP}

In this section, we review PCP \cite{wang2022probabilistic}. To this end, consider a parametric probabilistic predictor defined by a conditional distribution  $p(y|x,\theta)$ of target output $\mathbf{y}$ given input $\mathbf{x}=x$. For example, in regression, the distribution $p(y|x,\theta)$ may describe a Gaussian random variable $\mathbf{y}$ with mean and covariance dependent on input $x$ and parameter vector $\theta$; or it may describe a categorical random variable $\mathbf{y}$ with logit vector dependent on input $x$ and parameter vector $\theta$. The mentioned functions are typically implemented as neural networks with weight vector $\theta$ and input $x$. Model $p(y|x,\theta)$ can be trained using standard tools from machine learning, yielding a trained model parameter vector $\theta_{\mathcal{D}^\text{tr}}$ \cite{simeone2022machine, shalev2014understanding}. 

PCP constructs a set predictor not directly as a function of the likelihood $p(y|x,\theta_{\mathcal{D}^\text{tr}})$ of the trained model, but rather as a function of a number of random predictions $\hat{\mathbf{y}}$ generated from the model $p(y|x,\theta_{\mathcal{D}^\text{tr}})$. As mentioned in  Sec.~\ref{sec:intro}, the motivation of PCP is to obtain more flexible set predictors that can describe disjoint error bars \cite[Fig. 1]{wang2022probabilistic}. An illustration of PCP can be found {\color{black} in  Fig.~\ref{fig:qcp_scoring}(c)--(d).} 

To elaborate, given test input $x$ and trained model $p(y|x,\theta_{\mathcal{D}^\text{tr}})$, PCP generates $M$ i.i.d. predictions $\hat{\mathbf{y}}^{1:M} = \{ \hat{\mathbf{y}}^m \}_{m=1}^M$, with each sample drawn from the model as
\begin{align}
    \label{eq:sampling_classical}
    \hat{\mathbf{y}}^m \sim p(y|x,\theta_{\mathcal{D}^\text{tr}}). 
\end{align} It then calibrates these predictions by using the calibration set. To this end, given the calibration set $\mathcal{D}^\text{cal}$, in an \emph{offline} phase, for each calibration example $z[i]\in \mathcal{D}^\text{cal}$, PCP generates $M$ i.i.d. random predictions $\hat{\mathbf{y}}^{1:M}[i] = \{ \hat{\mathbf{y}}^m[i] \}_{m=1}^M$, with each sample obtained from the model as $\hat{\mathbf{y}}^m[i]\sim p(y|x[i],\theta_{\mathcal{D}^\text{tr}})$.

Like CP, PCP also relies on the use of a  scoring function that evaluates the loss of the trained model on each data point. Unlike CP, the scoring function of PCP is not a function of a single, deterministic, prediction $\hat{y}$, but rather of the $M$ random predictions $\hat{\mathbf{y}}^{1:M}$ generated i.i.d. from the model $p(y|x,\theta_{\mathcal{D}^\text{tr}})$ for the given test input $x$. Accordingly, we write as $s((x,y)|\hat{\mathbf{y}}^{1:M})$ the scoring function, which measures the loss obtained by the trained model on an example $z=(x,y)$ based on the random predictions $\hat{\mathbf{y}}^{1:M}$. 

Reference \cite{wang2022probabilistic} proposed the scoring function
\begin{align}
    \label{eq:scoring_pcp}
    s(z=(x,y)|\hat{\mathbf{y}}^{1:M}) = \min_{m \in \{1,...,M\}}|| y - \hat{\mathbf{y}}^m||.
\end{align}This loss metric uses the best among the $M$ random predictions $\hat{\mathbf{y}}^m$ to evaluate the score on example $z$ as  the Euclidean distance $|| y - \hat{\mathbf{y}}^m ||$ between the best prediction $\hat{\mathbf{y}}^m$ and the output $y$ \cite{wang2022probabilistic}.

Having defined the scoring function as in  \eqref{eq:scoring_pcp}, PCP obtains a set predictor  \eqref{eq:classical_CP_set_predictor} as for CP by replacing the scoring function  $s(z|\theta)$, e.g., the quadratic loss $(y-f(x|\theta))^2$, with the scoring function \eqref{eq:scoring_pcp} for both test and calibration pairs. This yields the set predictor
\begin{align}
\label{eq:prob_CP_set_predictor}
\mathbf{\Gamma}_M(x|\mathcal{D}^\text{cal},\theta_{\mathcal{D}^\text{tr}}) &=\Big\{ y' \in \mathcal{Y} \hspace{-0.1cm}\: : \hspace{-0.1cm}\: {s}((x,y')|\hat{\mathbf{y}}^{1:M})  \:\leq\:  Q_{1-\alpha} \Big(\big\{ {s}( z[i]|\hat{\mathbf{y}}^{1:M}[i]  )\big\}_{i=1}^{ |\mathcal{D}^\text{cal}| } \Big)\Big\}.   
 \end{align} When evaluated using the scoring function (\ref{eq:scoring_pcp}), this results in a generally disjoint set of intervals for scalar variables (and circles for two-dimensional variables) \cite{wang2022probabilistic}.

As long as the examples in the calibration set $\mathbfcal{D}^\text{cal}$ and the test example $\mathbf{z}$  i.i.d. random variables, the PCP set predictor \eqref{eq:prob_CP_set_predictor} is well calibrated as formalized next. 

\begin{theorem}[Calibration of PCP \cite{wang2022probabilistic}]\label{ther:pcp} 
Assuming that calibration data $\mathbfcal{D}^\text{cal}$ and test data $\mathbf{z}$ are i.i.d., with sampling procedure following \eqref{eq:sampling_classical}, for any miscoverage level $\alpha \in (0,1)$ and for any trained model $\theta_{\mathcal{D}^{\text{tr}}}$, the PCP set predictor \eqref{eq:prob_CP_set_predictor} satisfies the inequality \begin{align}
    \label{eq:validity_pcp}
    \Pr( \mathbf{y} \in \mathbf{\Gamma}_M(\mathbf{x}|\mathbfcal{D}^\text{cal},\theta_{\mathcal{D}^\text{tr}}) ) \geq 1-\alpha,
\end{align}with probability taken over the joint distribution of the test data $\mathbf{z}$, of the calibration data set $\mathbfcal{D}^\text{cal}$, and also over the independent random predictions $\hat{\mathbf{y}}^{1:M}$ and $\{\hat{\mathbf{y}}^{1:M}[i]\}_{i=1}^{|\mathcal{D}^\text{cal}|}$ produced by the model for test and calibration points.
\end{theorem}

A proof of Theorem~\ref{ther:pcp} is provided for completeness in the next section.

\section{Proof for the Theoretical Calibration Guarantees of QCP} 
\label{appendix:proofs}
In this section, we provide a unified proof for conventional CP (Sec.~\ref{sec:CCP}), PCP (Sec.~\ref{sec:probabilistic_CP}), and QCP (Sec.~\ref{sec:QCP}).  We first define \emph{finite exchangeability} as follows.
\begin{assumption}[Finite exchangeability \protect{\cite[Sec.~II]{caves2002unknown}}]\label{assum:classical} Calibration data set $\mathbfcal{D}^\text{cal}$ and a test data point $\mathbf{z}$ are finitely exchangeable random variables, i.e., the joint distribution $p(\mathcal{D}^\text{cal},z) = p(z[1],...,z[|\mathcal{D}^\text{cal}|],z)$ is invariant to any permutation of the variables $\{ \mathbf{z}[1],...,\mathbf{z}[|\mathcal{D}^\text{cal}|],\mathbf{z} \}$. Mathematically, we have the equality $p(z[1],...,z[|\mathcal{D}^\text{cal}|+1]) = p(z[\pi(1)],...,z[\pi(|\mathcal{D}^\text{cal}|+1)])$ with $z=z[|\mathcal{D}^\text{cal}|+1]$, for any permutation operator $\pi(\cdot)$. Note that the standard assumption of i.i.d. random variables satisfies finite exchangeability.
\end{assumption}

We  unify the expression for the scoring function as $t(z|\boldsymbol{\nu})$, where $\boldsymbol{\nu}$ is a context variable. This way, for deterministic CP (Sec.~\ref{sec:CCP}), we have $t(z|\boldsymbol{\nu}) = s(z|\theta_{\mathcal{D}^\text{tr}})$ with the context being deterministic; for PCP, we have $t(z|\boldsymbol{\nu}) = s(z| \hat{\mathbf{y}}^{1:M}) $ as in \eqref{eq:scoring_pcp}, with random variable $\boldsymbol{\nu} = \hat{\mathbf{y}}^{1:M}$ generated i.i.d. from the classical model $p(y|x,\theta_{\mathcal{D}^\text{tr}})$; and for QCP (Sec.~\ref{sec:QCP}), we have $t(z|\boldsymbol{\nu}) = s(z| \hat{\mathbf{y}}^{1:M} )$ with random variable $\boldsymbol{\nu} = \hat{\mathbf{y}}^{1:M}$ generated by the PQC with trained model $\theta_{\mathcal{D}^\text{tr}}$ via \eqref{eq:joint_disribution_M_shots}.

With this notation, the validity conditions proved by CP, PCP, and QCP can be expressed as \begin{align}
\label{eq:validity_condition_for_all_theorems}
&\Pr( \mathbf{y} \in {\Gamma}(\mathbf{x}|\mathbfcal{D}^\text{cal},\theta_{\mathcal{D}^\text{tr}}) )  \nonumber\\ &=  \Pr(  t(\mathbf{z}[N+1]|\boldsymbol{\nu}[N+1]) \leq Q_{1-\alpha}( \{ t(\mathbf{z}[i]|\boldsymbol{\nu}[i] ) \}_{i=1}^N  )   )   \nonumber\\&\geq 1-\alpha,
\end{align}
where $t(\mathbf{z}[N+1]|\boldsymbol{\nu}[N+1] )$ is the score for the test data $\mathbf{z}$. We  recall that, given a set of real numbers  $\{ s[1],...,s[N]\}$, the notation $Q_{1-\alpha}( \{ s[i]  \}_{i=1}^N  )$ represents the $\lceil (1-\alpha)(N+1)\rceil$-th smallest value in the set (for $\alpha \geq 1/(|\mathcal{D}^\text{cal}|+1)$). To prove the inequality \eqref{eq:validity_condition_for_all_theorems}, we introduce the following lemmas.

\begin{lemma}[Exchangeability lemma]\label{lemma:exchangeability}
Assume that $\mathbf{z}[1],...,\mathbf{z}[N+1] \in \mathcal{Z}$ are exchangeable random variables. Furthermore, assume that the joint distribution of random variables $(\mathbf{z}[1],\boldsymbol{\nu}[1]),  ..., (\mathbf{z}[N+1],\boldsymbol{\nu}[N+1])$ can be written as
\begin{align}
    p( (z[1],\nu[1]),...,(z[N+1],\nu[N+1])  )
    = \prod_{i=1}^{N+1} p(\nu[i]|z[i]) p(z[1],...,z[N+1])\label{eq:permutation_invariant_for_lemma}
    \end{align}for some conditional distribution $p(\nu[i]|z[i])$ that does not depend on $i=1,...,N+1$. Then, for any real-valued function $t(z|\nu)$, the random variables
\begin{align}
    t(\mathbf{z}[1]|\boldsymbol{\nu}[1]),...,t(\mathbf{z}[N+1]|\boldsymbol{\nu}[N+1])
\end{align}
are exchangeable.

\end{lemma}

\begin{proof}
The result follows directly from the permutation-invariance of the distribution (\ref{eq:permutation_invariant_for_lemma}). 
\end{proof}

\begin{lemma}[Quantile lemma \cite{tibshirani2019conformal}]\label{lemma:quantile_lemma}
If $\mathbf{s}[1],...,\mathbf{s}[N],\mathbf{s}[N+1]$ are exchangeable random variables, then for any $\alpha \in (0,1)$, the following inequality holds
\begin{align}
 \Pr(  \mathbf{s}[N+1] \leq Q_{1-\alpha}( \{ \mathbf{s}[i]  \}_{i=1}^N  )   )   \geq 1-\alpha.
\end{align}
\end{lemma}
\begin{proof}
Defining $Q^*_{1-\alpha}( \{ \mathbf{s}[i]  \}_{i=1}^{N+1}  )$ as the $\lceil (1-\alpha)(N+1)\rceil$-th smallest value in the set $\{ \mathbf{s}[1],...,\mathbf{s}[N],\mathbf{s}[N+1] \}$, we have the inequality 
\begin{align}
 \Pr(  \mathbf{s}[N+1] \leq Q^{*}_{1-\alpha}( \{ \mathbf{s}[i]  \}_{i=1}^{N+1}  )   )   \geq 1-\alpha,
\end{align} by the exchangeability of the random variables \cite{kuchibhotla2020exchangeability}. Furthermore, we have the following equivalence \cite[Sec.~A.1]{tibshirani2019conformal}
\begin{align}
    \label{eq:app_quantile_equiv}
    \mathbf{s}[N+1] > Q_{1-\alpha}( \{ \mathbf{s}[i]  \}_{i=1}^N  )
    \Leftrightarrow &\mathbf{s}[N+1] > Q^*_{1-\alpha}( \{ \mathbf{s}[i]  \}_{i=1}^{N+1}  ),
\end{align} which can be readily checked by noting that $\mathbf{s}[N+1]$ cannot be strictly larger than itself or of $\infty$.

\end{proof}

Combining Lemma~\ref{lemma:exchangeability} and Lemma~\ref{lemma:quantile_lemma} for the random variables $\mathbf{s}[i] = t(\mathbf{z}[i]|\boldsymbol{\nu}[i])$,  we obtain the desired condition \eqref{eq:validity_condition_for_all_theorems}.

\section{Calibration Guarantees over a Finite Number of Experiments}
\label{sec:finite_number_guarantee}
{\color{black} Throughout the paper}, we have described CP schemes that satisfy calibration conditions, namely \eqref{eq:relfirst}, \eqref{eq:validity_pcp}, and \eqref{eq:validity_qcp}, that are defined on average over the calibration and test data points. For PCP and QCP, the average is also taken with respect to the random predictions produced for calibration and test data points. In this section, we elaborate on the practical significance of this expectation.

Suppose that we run CP, PCP or QCP (Algorithm~\ref{alg:QCP}) over $K$ runs that use independent calibration and test data. What is the fraction of such runs that  meet the condition that the true output is included in the predictive set? Ideally, this fraction will be close to the desired target $1-\alpha$ with high probability. In fact, by the results in the previous sections and by the law of large numbers, as $K$ grows large, this fraction of ``successful'' experiments will tend to $1-\alpha$. What can be guaranteed for a finite number $K$ of experiments? In the following, we address this question for conventional CP first, and then for PCP and QCP.

\subsection{Conventional CP}
Following the setting described above, let us consider $K$ experiments, such that in each $k$-th experiment we draw calibration and test data from the joint distribution $p(\mathcal{D}^\text{cal},z)$, i.e., $(\mathbfcal{D}^\text{cal}_k, \mathbf{z}_k) \sim p(\mathcal{D}^\text{cal},z)$. For each experiment, we evaluate whether the prediction set \eqref{eq:classical_CP_set_predictor} produced by CP contains the true target label $\mathbf{y}_k$ or not. Accordingly, the fraction of ``successful'' experiments is given by \begin{align}
    \label{eq:empirical_cov_dcp}
    \hat{\mathbf{P}}= \frac{1}{K}\sum_{k=1}^K \mathbbm{1}\big(\mathbf{y}_k \in \Gamma(\mathbf{x}_k|\mathbfcal{D}_k^\text{cal},\theta_{\mathcal{D}^\text{tr}})\big),
\end{align} where $\mathbbm{1}(\cdot)$ is the indicator function ($\mathbbm{1}(\text{true})=1$ and $\mathbbm{1}(\text{false})=0$). To restate the question posed at the beginning of this section, given $K$, how large can we guarantee the success rate $\hat{P}$ to be?

By the exchangeability of calibration and test data (Assumption 1), which implies the exchangeability of the $|\mathcal{D}^\text{cal}|+1$ scores evaluated on calibration and test data (see Appendix~\ref{appendix:proofs}), the distribution of random variable $K\hat{\mathbf{P}}$ is given by the binomial
\begin{align}
    \label{eq:convergence_pmf}
    K\hat{\mathbf{P}} \sim \text{Binom}\bigg(K, \frac{\lceil (1-\alpha)(|\mathcal{D}^\text{cal}|+1) \rceil}{|\mathcal{D}^\text{cal}|+1}\bigg),
\end{align}
if ties between the $|\mathcal{D}^\text{cal}|+1$ scores occur with probability zero (see \cite{ lei2018distribution, tibshirani2019conformal} for other uses of this assumption). Note that  the distribution \eqref{eq:convergence_pmf} can be recovered from \cite[Sec.~C]{angelopoulos2021gentle} by setting the number of test points to be one (see also \cite{vovk2012conditional}).

This implies that the success rate $\hat{\mathbf{P}}$ is larger than $1-\alpha-\epsilon$ for any $\epsilon>0$ with probability \begin{align}
    \label{eq:convergence_theor}
    &\Pr(  \hat{\mathbf{P}} \geq 1-\alpha -\epsilon ) = I_{ \frac{\lceil (1-\alpha)(|\mathcal{D}^\text{cal}|+1) \rceil}{|\mathcal{D}^\text{cal}|+1}   } ( \lceil K(1-\alpha-\epsilon) \rceil, \lfloor K(\alpha+\epsilon)\rfloor+1 ),
\end{align}
with regularized incomplete beta function $I_x(a,b)=B(x;a,b)/B(a,b)$, where $B(x;a,b)=\int_{0}^x t^{a-1}(1-t)^{b-1}\mathrm{d}t$ and $B(a,b)=B(1;a,b)$ \cite{jowett1963relationship}.

Fig.~\ref{fig:convergence_cp_theory} shows the probability distribution of random variable $\hat{\mathbf{P}}$, as well as the probability (\ref{eq:convergence_theor})  as a function of tolerance level $\epsilon$ for different number of experiments $K=100,1000,10000$, given  $\alpha=0.1$ and $|\mathcal{D}^\text{cal}|=9$. The top figure confirms that, by the law of large numbers, the distribution of success rate $\hat{\mathbf{P}}$ concentrates around the value $1-\alpha=0.9$. The bottom figure can be used to identify a value of the \emph{backoff probability} $\epsilon$ that allows one to obtain finite-$K$ guarantees on the success rate $1-\alpha-\epsilon$. For instance,  when $K = 1000$, we observe that setting $\epsilon=0.03$ guarantees a success rate $\hat{\mathbf{P}}$ no smaller than $0.87$ with probability larger than $0.999$.


\begin{figure}
  \centering
  \includegraphics[width=0.48\textwidth]{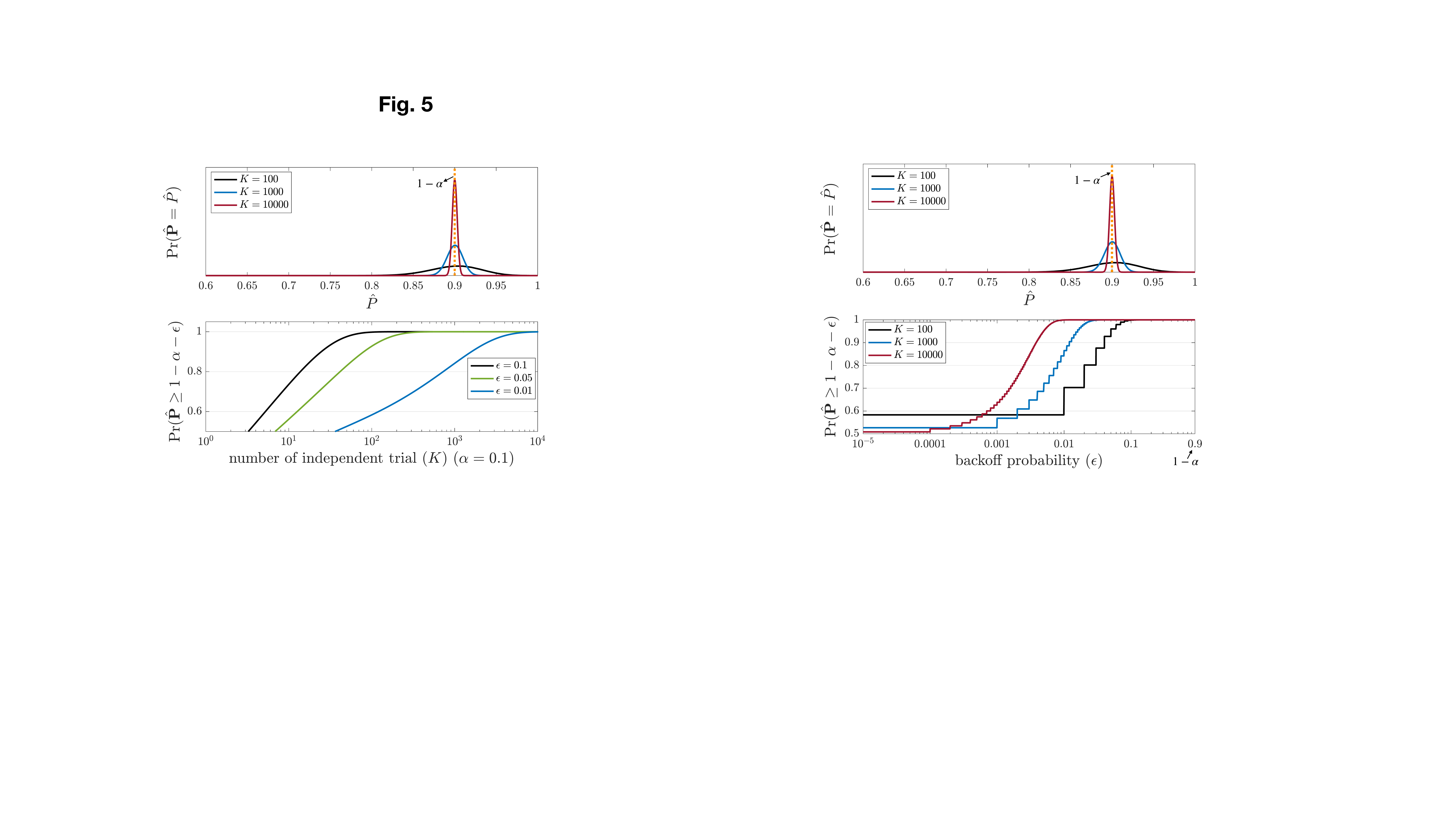}
  \caption{(top) Normalized density of the empirical coverage rate $\hat{\mathbf{P}}$ obtained from $K=100,1000,10000$ independent trials. (bottom) Probability that the empirical coverage rate $\hat{\mathbf{P}}$ satisfies the validity condition with tolerance level, or backoff probability, $\epsilon$ as a function of $\epsilon$, for $K=100,1000,10000$ independent trials  ($\alpha=0.1, |\mathcal{D}^\text{cal}|=9$).}
  \label{fig:convergence_cp_theory}
\end{figure}

\subsection{Probabilistic and Quantum CP}
In the case of PCP and QCP, each $k$-th experiment involves also the  predictions $\{\hat{\mathbf{y}}_k^{1:M}[i]\}_{i=1}^{|\mathcal{D}^\text{cal}|}$ for the calibration points $\mathbfcal{D}^\text{cal}_k$, and  $\hat{\mathbf{y}}_k^{1:M}$ for the test data $\mathbf{z}_k$, following either \eqref{eq:sampling_classical} (PCP) or \eqref{eq:single_measurement_distribution_noisy} (QCP). Despite the presence of the additional randomness due to the stochastic predictions, the finite-$K$ guarantees \eqref{eq:convergence_pmf}-\eqref{eq:convergence_theor} still hold for the fraction of ``successful'' experiments
\begin{align}
    \label{eq:empirical_cov_qcp}
    \hat{\mathbf{P}}_M := \frac{1}{K}\sum_{k=1}^K \mathbbm{1}\big(\mathbf{y}_k \in \mathbf{\Gamma}_M(\mathbf{x}_k|\mathbfcal{D}_k^\text{cal},\theta_{\mathcal{D}^\text{tr}})\big)
\end{align}
for PCP and QCP if ties between the $|\mathcal{D}^\text{cal}|+1$ scores occur with probability zero. This is because the $|\mathcal{D}^\text{cal}|+1$ scores for calibration and test data are exchangeable also for PCP and QCP due to the exchangeability of calibration and test data (Assumption~\ref{assum:classical}), and to the independence of $M$ predictions \eqref{eq:sampling_classical}, \eqref{eq:single_measurement_distribution_noisy} for distinct inputs.

\section{Additional Experiments} \label{appendix:additional_experiments}
{\color{black}
\subsection{Impact of the Choice of Parameter $\lowercase{k}$ in the Scoring Function (18)}
\begin{figure}
  \centering
  \includegraphics[width=0.9\textwidth]{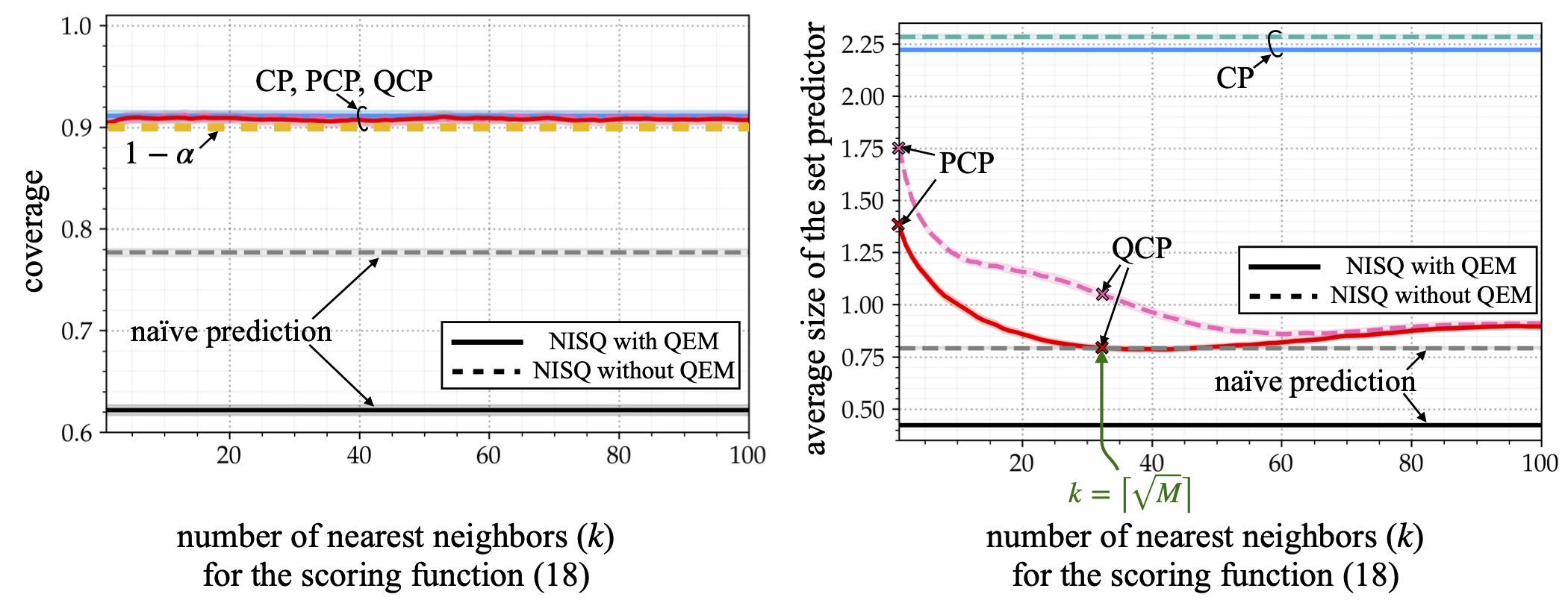}
      \caption{{\color{black} Density estimation with a strongly bimodal ground-truth Gaussian distribution: Coverage and average size of the set predictors as a function of $k$ given $|\mathcal{D}|=20$ available training samples. Training is done on a classical simulator,  while testing is implemented on  \texttt{imbq\_quito} NISQ device, with or without M3 QEM \cite{nation2021scalable}. The shaded areas correspond to confidence intervals covering 95\% of the realized values.} } \label{fig:per_k_exp}
\end{figure}

To elaborate on the impact of the choice of parameter $k$ {\color{black} in the scoring function \eqref{eq:pcp_scoring_from_generalized_with_k}  when used in conjunction with QCP,} we plot in Fig.~\ref{fig:per_k_exp} the coverage and average size of the QCP set predictor for the density learning problem (see Sec.~\ref{subsec:unsup_learning}) as a function of $k$ given availability of $M=1000$ measurements. Recall that $k=1$ corresponds to the choice of scoring function assumed in PCP \cite{wang2022probabilistic}, while the selection $k=\lceil \sqrt{M} \rceil$ ensures consistency of the $k$-NN density estimator as summarized in Sec.~\ref{subsubsec:quantum_scoring}. From Fig.~\ref{fig:per_k_exp} we conclude  that the proposed scoring function \eqref{eq:pcp_scoring_from_generalized} with the theoretically motivated choice $k=\lceil \sqrt{M} \rceil$  achieves nearly minimal average predicted set size, decreasing the set size by $57.25\%$ as compared to the case $k=1$, {\color{black}which PCP assumes. }

\subsection{QCP with PQC Trained in the Presence of Quantum Hardware Noise}
In order to further verify that the reliability guarantees of QCP in Theorem \ref{ther:qcp} hold irrespective of the quality of the trained PQC, even trained in the presence of quantum hardware noise, in Fig. \ref{fig:tr_on_ibmq}, we plot the coverage and average size of the set predictors for the problem of density estimation using a PQC trained on \texttt{imbq\_quito} NISQ device with or without QEM. The other settings are same as in Fig. \ref{fig:density_learning_intro}. QCP is observed to guarantee  reliability also for the model trained on the quantum computer. This is in contrast to the na\"ive set predictor, which only covers $40\%$ of the support, falling far short of the target coverage level $90\%$.
\begin{figure}
  \centering
  \includegraphics[width=0.9\textwidth]{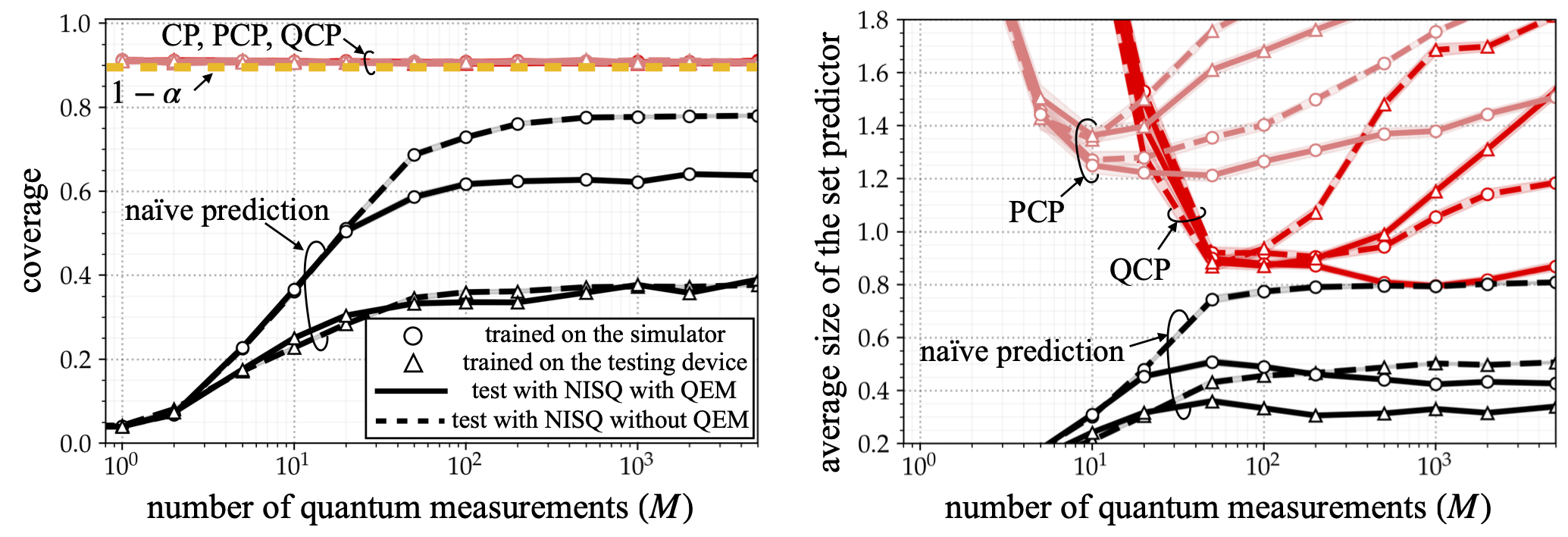}
  \caption{{\color{black}Density estimation with a strongly bimodal ground-truth Gaussian distribution: Coverage and average size of the set predictors as a function of number $M$ of quantum measurements given $|\mathcal{D}|=20$ available training samples.  Training is done either on a classical simulator or on a \texttt{imbq\_quito} NISQ device with or without M3 QEM \cite{nation2021scalable}. When trained on the quantum device, the same QEM strategy is applied during the testing phase.  The shaded areas correspond to confidence intervals covering 95\% of the realized values.} }
\label{fig:tr_on_ibmq}
\vspace{-0.5cm}
\end{figure}

\subsection{Impact of the Temperature in Quantum Data Classification}\label{sec:exp_qc_app}
In this subsection, we provide additional experimental results to study the impact of temperature $T$ in the Gibbs state $\rho(y)$ (Sec.~\ref{sec:exp_qc}), as well as of the number of measurements $M$, for quantum data classification.

\begin{figure}
  \centering
  \includegraphics[width=0.9\textwidth]{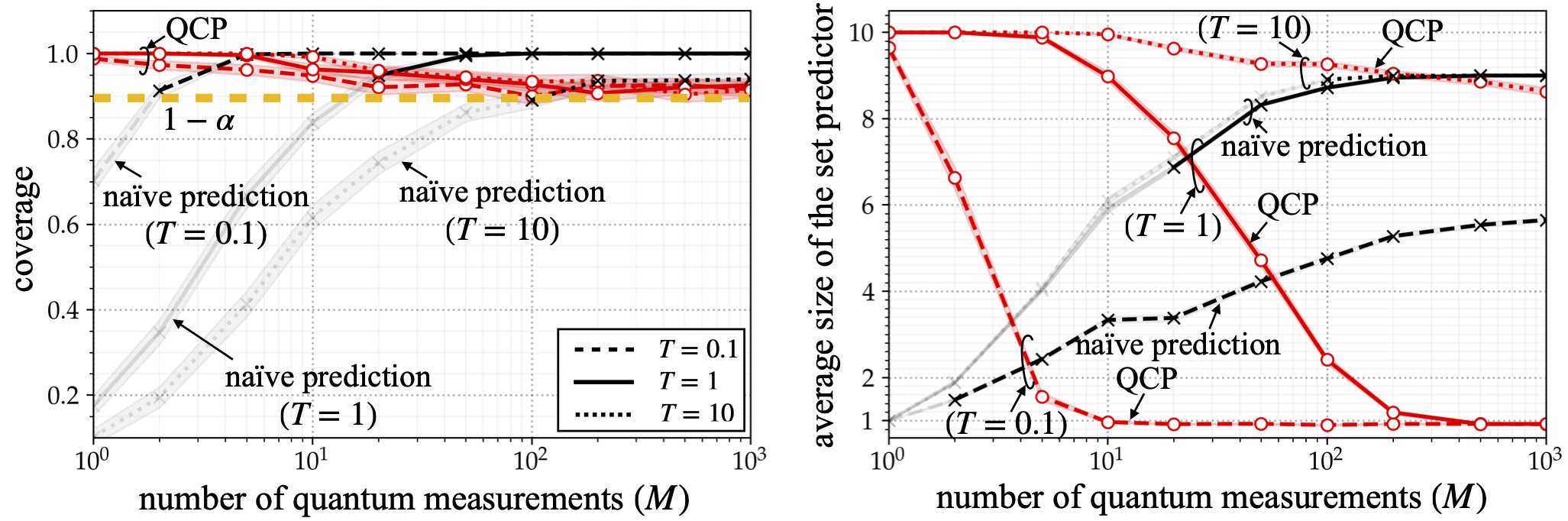}
  \caption{{\color{black}Quantum data classification: Coverage and average size of the set predictors as a function of the number $M$ of quantum measurements given $|\mathcal{D}^\text{cal}|=10$ calibration samples. The ten possible density matrices to be classified are generated as $\rho(y) = e^{-H(y)/T}/\text{Tr}(e^{-H(y)/T})$ with temperature $T>0$, where the Hamiltonian matrices $H(y)$ are independently generated so as to ensure a sparsity level of $0.2$ at temperature $T=1$ as in \cite{ahmed2021classification}. Pretty good measurements detector is adopted \cite{hausladen1994pretty}. The shaded areas correspond to confidence intervals covering 95\% of the realized values.  The results are averaged over $1000$ experiments, and transparent lines are used to highlight regimes in which the  set predictors do not meet the coverage level $1-\alpha=0.9$.}} 
\label{fig:qc_exp_main}
\end{figure}

In Fig.~\ref{fig:qc_exp_main}, we plot coverage and average size of the set predictor as a function of number of measurements $M$ for QCP and for the na\"ive predictor.  We consider no drift, i.e., $\tau=\infty$ in \eqref{eq:decoherence_evolution}. Note that an increased temperature $T$ makes the classification problem more challenging since all the $C$ density matrices $\{\rho(y)\}_{y=1}^C$ become increasingly close to the maximally mixed state. As per our theoretical results, QCP always achieves coverage no smaller than the predetermined level $1-\alpha=0.9$, while the na\"ive prediction fails to achieve validity unless it is supplied a sufficiently large number of quantum measurements $M$, i.e., $M\geq 100$ for temperature $T=10$. In the regime of many shots, i.e., for $M \geq 100$, although both na\"ive and QCP set predictors are valid, the  na\"ive set predictor tends to be extremely conservative, yielding predicted set that include $9$ out of $10$ labels for $T=1$, while the set predictions output by QCP include on average a single class.

}